\documentclass[submission,copyright,creativecommons]{eptcs}
\providecommand{\event}{QPL 2022} 
\usepackage{scrextend}
\usepackage[iso]{datetime}
\usepackage{graphicx}
\usepackage[utf8]{inputenc}
\usepackage[british]{babel}
\usepackage{caption}
\usepackage{float}
\usepackage{newfloat}
\usepackage{framed}
\usepackage{amssymb}
\usepackage{amsmath}
\usepackage{mathtools}
\usepackage{braket}
\usepackage{color}
\usepackage{titlesec}
\usepackage{subcaption}
\usepackage[numbers,comma]{natbib}
\usepackage{algorithm}
\usepackage[noend]{algpseudocode}
\usepackage{bbm}
\usepackage{anyfontsize}
\usepackage{mathrsfs}
\usepackage[linewidth=1pt]{mdframed}

\usepackage{stmaryrd}
\SetSymbolFont{stmry}{bold}{U}{stmry}{m}{n}

\usepackage{bussproofs}
\usepackage{listingsutf8}
\usepackage[normalem]{ulem}
\usepackage{proof}
\usepackage{xfrac}
\usepackage[thmmarks, amsmath, thref]{ntheorem}
\usepackage{pdfcomment}
\usepackage{import}
\usepackage[frozencache]{minted}

\titleformat{\subsection}[runin]
       {\normalfont\bfseries}
       {\thesubsection}
       {0.5em}
       {}
       [.]

\usepackage{tikzit}
\tikzstyle{env}=[copoint,regular polygon rotate=0,minimum width=0.2cm, fill=black]

\tikzstyle{probs}=[shape=semicircle,fill=white,draw=black,shape border rotate=180,minimum width=1.2cm]

\tikzstyle{every picture}=[baseline=-0.25em,scale=0.5]
\tikzstyle{dotpic}=[] 
\tikzstyle{diredges}=[every to/.style={diredge}]
\tikzstyle{math matrix}=[matrix of math nodes,left delimiter=(,right delimiter=),inner sep=2pt,column sep=1em,row sep=0.5em,nodes={inner sep=0pt},text height=1.5ex, text depth=0.25ex]

\tikzstyle{gs edge}=[]
\tikzstyle{gs double edge}=[double,shorten <=-1mm,shorten >=-1mm,double distance=2pt]

\tikzstyle{inline text}=[text height=1.5ex, text depth=0.25ex,yshift=0.5mm]
\tikzstyle{label}=[font=\footnotesize,text height=1.5ex, text depth=0.25ex,yshift=0.5mm]
\tikzstyle{left label}=[label,anchor=east,xshift=1.5mm]
\tikzstyle{right label}=[label,anchor=west,xshift=-1.5mm]

\tikzstyle{braceedge}=[decorate,decoration={brace,amplitude=2mm,raise=-1mm}]
\tikzstyle{small braceedge}=[decorate,decoration={brace,amplitude=1mm,raise=-1mm}]

\tikzstyle{doubled}=[line width=1.6pt] 
\tikzstyle{boldedge}=[doubled,shorten <=-0.17mm,shorten >=-0.17mm]
\tikzstyle{boldedgegray}=[doubled,gray,shorten <=-0.17mm,shorten >=-0.17mm]
\tikzstyle{singleedgegray}=[gray]

\tikzstyle{semidoubled}=[line width=1.4pt] 
\tikzstyle{semiboldedgegray}=[semidoubled,gray,shorten <=-0.17mm,shorten >=-0.17mm]

\tikzstyle{boxedge}=[semiboldedgegray]

\tikzstyle{boldedgedashed}=[very thick,dashed,shorten <=-0.17mm,shorten >=-0.17mm]
\tikzstyle{vboldedgedashed}=[doubled,dashed,shorten <=-0.17mm,shorten >=-0.17mm]
\tikzstyle{left hook arrow}=[left hook-latex]
\tikzstyle{right hook arrow}=[right hook-latex]
\tikzstyle{sembracket}=[line width=0.5pt,shorten <=-0.07mm,shorten >=-0.07mm]

\tikzstyle{causal edge}=[->,thick,gray]
\tikzstyle{causal nondir}=[thick,gray]
\tikzstyle{timeline}=[thick,gray, dashed]

\tikzstyle{cedge}=[<->,thick,gray!70!white]

\tikzstyle{empty diagram}=[draw=gray!40!white,dashed,shape=rectangle,minimum width=1cm,minimum height=1cm]
\tikzstyle{empty diagram small}=[draw=gray!50!white,dashed,shape=rectangle,minimum width=0.6cm,minimum height=0.5cm]

\tikzstyle{dot}=[inner sep=0mm,minimum width=2mm,minimum height=2mm,draw,shape=circle]  
\tikzstyle{Wsquare}=[white dot, shape=regular polygon, rounded corners=0.8 mm, minimum size=3.3 mm, regular polygon sides=3, outer sep=-0.2mm]
\tikzstyle{Wsquareadj}=[white dot, shape=regular polygon, rounded corners=0.8 mm, minimum size=3.3 mm, regular polygon sides=3, outer sep=-0.2mm, regular polygon rotate=180]
\tikzstyle{ddot}=[inner sep=0mm, doubled, minimum width=2.5mm,minimum height=2.5mm,draw,shape=circle]

\tikzstyle{black dot}=[dot,fill=black]
\tikzstyle{white dot}=[dot,fill=white,,text depth=-0.2mm]
\tikzstyle{white Wsquare}=[Wsquare,fill=white,,text depth=-0.2mm]
\tikzstyle{white Wsquareadj}=[Wsquareadj,fill=white,,text depth=-0.2mm]
\tikzstyle{green dot}=[white dot] 
\tikzstyle{gray dot}=[dot,fill=gray!40!white,,text depth=-0.2mm]
\tikzstyle{red dot}=[gray dot] 

\tikzstyle{Z}=[white dot]
\tikzstyle{X}=[gray dot]
\tikzstyle{simple}=[]

\tikzstyle{black ddot}=[ddot,fill=black]
\tikzstyle{white ddot}=[ddot,fill=white]
\tikzstyle{gray ddot}=[ddot,fill=gray!40!white]

\tikzstyle{gray edge}=[gray!40!white]

\tikzstyle{small dot}=[inner sep=1pt,minimum width=0pt,minimum height=0pt,draw,shape=circle]

\tikzstyle{small black dot}=[small dot,fill=black]
\tikzstyle{small white dot}=[small dot,fill=white]
\tikzstyle{small gray dot}=[small dot,fill=gray!40!white]
\tikzstyle{special dot} = [small white dot]

\tikzstyle{mbqc dot}=[small black dot]
\tikzstyle{mbqc input dot}=[small white dot]
\tikzstyle{mbqc output dot}=[small gray dot]

\tikzstyle{causal dot}=[inner sep=0.4mm,minimum width=0pt,minimum height=0pt,draw=white,shape=circle,fill=gray!40!white]

\tikzstyle{phase dimensions}=[minimum size=5mm,font=\footnotesize,rectangle,rounded corners=2mm,inner sep=0.2mm,outer sep=-2mm,scale=0.8]
\tikzstyle{dphase dimensions}=[minimum size=5mm,font=\footnotesize,rectangle,rounded corners=2.5mm,inner sep=0.2mm,outer sep=-2mm]

\tikzstyle{white phase dot}=[dot,fill=white,phase dimensions]
\tikzstyle{white phase ddot}=[ddot,fill=white,dphase dimensions]

\tikzstyle{white rect ddot}=[draw=black,fill=white,doubled,minimum size=5mm,font=\footnotesize,rectangle,rounded corners=2.5mm,inner sep=0.2mm]
\tikzstyle{gray rect ddot}=[draw=black,fill=gray!40!white,doubled,minimum size=6mm,font=\footnotesize,rectangle,rounded corners=3mm]

\tikzstyle{gray phase dot}=[dot,fill=gray!40!white,phase dimensions]
\tikzstyle{gray phase ddot}=[ddot,fill=gray!40!white,dphase dimensions]
\tikzstyle{grey phase dot}=[gray phase dot]
\tikzstyle{grey phase ddot}=[gray phase ddot]

\tikzstyle{small phase dimensions}=[minimum size=4mm,font=\tiny,rectangle,rounded corners=2mm,inner sep=0.2mm,outer sep=-2mm]
\tikzstyle{small dphase dimensions}=[minimum size=4mm,font=\tiny,rectangle,rounded corners=2mm,inner sep=0.2mm,outer sep=-2mm]

\tikzstyle{small gray phase dot}=[dot,fill=gray!40!white,small phase dimensions]
\tikzstyle{small gray phase ddot}=[ddot,fill=gray!40!white,small dphase dimensions]

\tikzstyle{small map}=[draw,shape=rectangle,minimum height=4mm,minimum width=4mm,fill=white]

\tikzstyle{cnot}=[fill=white,shape=circle,inner sep=-1.4pt]

\tikzstyle{asym hadamard}=[fill=white,draw,shape=NEbox,inner sep=0.6mm,font=\footnotesize,minimum height=4mm]
\tikzstyle{asym hadamard conj}=[fill=white,draw,shape=NWbox,inner sep=0.6mm,font=\footnotesize,minimum height=4mm]
\tikzstyle{asym hadamard dag}=[fill=white,draw,shape=SEbox,inner sep=0.6mm,font=\footnotesize,minimum height=4mm]

\tikzstyle{hadamard}=[fill=white,draw,inner sep=0.6mm,font=\footnotesize,minimum height=4mm,minimum width=4mm]
\tikzstyle{small hadamard}=[fill=white,draw,inner sep=0.6mm,minimum height=1.5mm,minimum width=1.5mm]
\tikzstyle{small hadamard rotate}=[small hadamard,rotate=45]
\tikzstyle{dhadamard}=[hadamard,doubled]
\tikzstyle{small dhadamard}=[small hadamard,doubled]
\tikzstyle{small dhadamard rotate}=[small hadamard rotate,doubled]
\tikzstyle{antipode}=[white dot,inner sep=0.3mm,font=\footnotesize]

\tikzstyle{scalar}=[diamond,draw,inner sep=0.5pt,font=\small]
\tikzstyle{dscalar}=[diamond,doubled, draw,inner sep=0.5pt,font=\small]

\tikzstyle{small box}=[rectangle,inline text,fill=white,draw,minimum height=5mm,yshift=-0.5mm,minimum width=5mm,font=\small]
\tikzstyle{small gray box}=[small box,fill=gray!30]
\tikzstyle{medium box}=[rectangle,inline text,fill=white,draw,minimum height=5mm,yshift=-0.5mm,minimum width=10mm,font=\small]
\tikzstyle{square box}=[small box] 
\tikzstyle{medium gray box}=[small box,fill=gray!30]
\tikzstyle{semilarge box}=[rectangle,inline text,fill=white,draw,minimum height=5mm,yshift=-0.5mm,minimum width=12.5mm,font=\small]
\tikzstyle{large box}=[rectangle,inline text,fill=white,draw,minimum height=5mm,yshift=-0.5mm,minimum width=15mm,font=\small]
\tikzstyle{large gray box}=[small box,fill=gray!30]

\tikzstyle{Bayes box}=[rectangle,fill=black,draw, minimum height=3mm, minimum width=3mm]

\tikzstyle{gray square point}=[small box,fill=gray!50]

\tikzstyle{dphase box white}=[dhadamard]
\tikzstyle{dphase box gray}=[dhadamard,fill=gray!50!white]
\tikzstyle{phase box white}=[hadamard]
\tikzstyle{phase box gray}=[hadamard,fill=gray!50!white]

\tikzstyle{point}=[regular polygon,regular polygon sides=3,draw,scale=0.75,inner sep=-0.5pt,minimum width=9mm,fill=white,regular polygon rotate=180]
\tikzstyle{copoint}=[regular polygon,regular polygon sides=3,draw,scale=0.75,inner sep=-0.5pt,minimum width=9mm,fill=white]
\tikzstyle{dpoint}=[point,doubled]
\tikzstyle{dcopoint}=[copoint,doubled]

\tikzstyle{wide copoint}=[fill=white,draw,shape=isosceles triangle,shape border rotate=90,isosceles triangle stretches=true,inner sep=0pt,minimum width=1.5cm,minimum height=6.12mm]
\tikzstyle{wide point}=[fill=white,draw,shape=isosceles triangle,shape border rotate=-90,isosceles triangle stretches=true,inner sep=0pt,minimum width=1.5cm,minimum height=6.12mm,yshift=-0.0mm]
\tikzstyle{wide point plus}=[fill=white,draw,shape=isosceles triangle,shape border rotate=-90,isosceles triangle stretches=true,inner sep=0pt,minimum width=1.74cm,minimum height=7mm,yshift=-0.0mm]

\tikzstyle{wide dpoint}=[fill=white,doubled,draw,shape=isosceles triangle,shape border rotate=-90,isosceles triangle stretches=true,inner sep=0pt,minimum width=1.5cm,minimum height=6.12mm,yshift=-0.0mm]

\tikzstyle{tinypoint}=[regular polygon,regular polygon sides=3,draw,scale=0.55,inner sep=-0.15pt,minimum width=6mm,fill=white,regular polygon rotate=180] 

\tikzstyle{white point}=[point]
\tikzstyle{white dpoint}=[dpoint]
\tikzstyle{green point}=[white point] 
\tikzstyle{white copoint}=[copoint]
\tikzstyle{gray point}=[point,fill=gray!40!white]
\tikzstyle{gray dpoint}=[gray point,doubled]
\tikzstyle{red point}=[gray point] 
\tikzstyle{gray copoint}=[copoint,fill=gray!40!white]
\tikzstyle{gray dcopoint}=[gray copoint,doubled]

\tikzstyle{white point guide}=[regular polygon,regular polygon sides=3,font=\scriptsize,draw,scale=0.65,inner sep=-0.5pt,minimum width=9mm,fill=white,regular polygon rotate=180]

\tikzstyle{black point}=[point,fill=black,font=\color{white}]
\tikzstyle{black copoint}=[copoint,fill=black,font=\color{white}]

\tikzstyle{tiny gray point}=[tinypoint,fill=gray!40!white]

\tikzstyle{diredge}=[->]
\tikzstyle{ddiredge}=[<->]
\tikzstyle{rdiredge}=[<-]
\tikzstyle{thickdiredge}=[->, very thick]
\tikzstyle{pointer edge}=[->,very thick,gray]
\tikzstyle{pointer edge part}=[very thick,gray]
\tikzstyle{dashed edge}=[dashed]
\tikzstyle{thick dashed edge}=[very thick,dashed]
\tikzstyle{thick gray dashed edge}=[thick dashed edge,gray!40]
\tikzstyle{thick map edge}=[very thick,|->]

\makeatletter
\newcommand{\boxshape}[3]{%
\pgfdeclareshape{#1}{
\inheritsavedanchors[from=rectangle] 
\inheritanchorborder[from=rectangle]
\inheritanchor[from=rectangle]{center}
\inheritanchor[from=rectangle]{north}
\inheritanchor[from=rectangle]{south}
\inheritanchor[from=rectangle]{west}
\inheritanchor[from=rectangle]{east}
\backgroundpath{
\southwest \pgf@xa=\pgf@x \pgf@ya=\pgf@y
\northeast \pgf@xb=\pgf@x \pgf@yb=\pgf@y

\@tempdima=#2
\@tempdimb=#3

\pgfpathmoveto{\pgfpoint{\pgf@xa - 5pt + \@tempdima}{\pgf@ya}}
\pgfpathlineto{\pgfpoint{\pgf@xa - 5pt - \@tempdima}{\pgf@yb}}
\pgfpathlineto{\pgfpoint{\pgf@xb + 5pt + \@tempdimb}{\pgf@yb}}
\pgfpathlineto{\pgfpoint{\pgf@xb + 5pt - \@tempdimb}{\pgf@ya}}
\pgfpathlineto{\pgfpoint{\pgf@xa - 5pt + \@tempdima}{\pgf@ya}}
\pgfpathclose
}
}}

\boxshape{NEbox}{0pt}{5pt}
\boxshape{SEbox}{0pt}{-5pt}
\boxshape{NWbox}{5pt}{0pt}
\boxshape{SWbox}{-5pt}{0pt}
\boxshape{EBox}{-3pt}{3pt}
\boxshape{WBox}{3pt}{-3pt}
\makeatother

\tikzstyle{cloud}=[shape=cloud,draw,minimum width=1.5cm,minimum height=1.5cm]

\tikzstyle{map}=[draw,shape=NEbox,inner sep=2pt,minimum height=6mm,fill=white]
\tikzstyle{dashedmap}=[draw,dashed,shape=NEbox,inner sep=2pt,minimum height=6mm,fill=white]
\tikzstyle{mapdag}=[draw,shape=SEbox,inner sep=2pt,minimum height=6mm,fill=white]
\tikzstyle{mapadj}=[draw,shape=SEbox,inner sep=2pt,minimum height=6mm,fill=white]
\tikzstyle{maptrans}=[draw,shape=SWbox,inner sep=2pt,minimum height=6mm,fill=white]
\tikzstyle{mapconj}=[draw,shape=NWbox,inner sep=2pt,minimum height=6mm,fill=white]

\tikzstyle{medium map}=[draw,shape=NEbox,inner sep=2pt,minimum height=6mm,fill=white,minimum width=7mm]
\tikzstyle{medium map dag}=[draw,shape=SEbox,inner sep=2pt,minimum height=6mm,fill=white,minimum width=7mm]
\tikzstyle{medium map adj}=[draw,shape=SEbox,inner sep=2pt,minimum height=6mm,fill=white,minimum width=7mm]
\tikzstyle{medium map trans}=[draw,shape=SWbox,inner sep=2pt,minimum height=6mm,fill=white,minimum width=7mm]
\tikzstyle{medium map conj}=[draw,shape=NWbox,inner sep=2pt,minimum height=6mm,fill=white,minimum width=7mm]
\tikzstyle{semilarge map}=[draw,shape=NEbox,inner sep=2pt,minimum height=6mm,fill=white,minimum width=9.5mm]
\tikzstyle{semilarge map trans}=[draw,shape=SWbox,inner sep=2pt,minimum height=6mm,fill=white,minimum width=9.5mm]
\tikzstyle{semilarge map adj}=[draw,shape=SEbox,inner sep=2pt,minimum height=6mm,fill=white,minimum width=9.5mm]
\tikzstyle{semilarge map dag}=[draw,shape=SEbox,inner sep=2pt,minimum height=6mm,fill=white,minimum width=9.5mm]
\tikzstyle{semilarge map conj}=[draw,shape=NWbox,inner sep=2pt,minimum height=6mm,fill=white,minimum width=9.5mm]
\tikzstyle{large map}=[draw,shape=NEbox,inner sep=2pt,minimum height=6mm,fill=white,minimum width=12mm]
\tikzstyle{large map conj}=[draw,shape=NWbox,inner sep=2pt,minimum height=6mm,fill=white,minimum width=12mm]
\tikzstyle{very large map}=[draw,shape=NEbox,inner sep=2pt,minimum height=6mm,fill=white,minimum width=17mm]

\tikzstyle{medium dmap}=[draw,doubled,shape=NEbox,inner sep=2pt,minimum height=6mm,fill=white,minimum width=7mm]
\tikzstyle{medium dmap dag}=[draw,doubled,shape=SEbox,inner sep=2pt,minimum height=6mm,fill=white,minimum width=7mm]
\tikzstyle{medium dmap adj}=[draw,doubled,shape=SEbox,inner sep=2pt,minimum height=6mm,fill=white,minimum width=7mm]
\tikzstyle{medium dmap trans}=[draw,doubled,shape=SWbox,inner sep=2pt,minimum height=6mm,fill=white,minimum width=7mm]
\tikzstyle{medium dmap conj}=[draw,doubled,shape=NWbox,inner sep=2pt,minimum height=6mm,fill=white,minimum width=7mm]
\tikzstyle{semilarge dmap}=[draw,doubled,shape=NEbox,inner sep=2pt,minimum height=6mm,fill=white,minimum width=9.5mm]
\tikzstyle{semilarge dmap trans}=[draw,doubled,shape=SWbox,inner sep=2pt,minimum height=6mm,fill=white,minimum width=9.5mm]
\tikzstyle{semilarge dmap adj}=[draw,doubled,shape=SEbox,inner sep=2pt,minimum height=6mm,fill=white,minimum width=9.5mm]
\tikzstyle{semilarge dmap dag}=[draw,doubled,shape=SEbox,inner sep=2pt,minimum height=6mm,fill=white,minimum width=9.5mm]
\tikzstyle{semilarge dmap conj}=[draw,doubled,shape=NWbox,inner sep=2pt,minimum height=6mm,fill=white,minimum width=9.5mm]
\tikzstyle{large dmap}=[draw,doubled,shape=NEbox,inner sep=2pt,minimum height=6mm,fill=white,minimum width=12mm]
\tikzstyle{large dmap conj}=[draw,doubled,shape=NWbox,inner sep=2pt,minimum height=6mm,fill=white,minimum width=12mm]
\tikzstyle{large dmap trans}=[draw,doubled,shape=SWbox,inner sep=2pt,minimum height=6mm,fill=white,minimum width=12mm]
\tikzstyle{large dmap adj}=[draw,doubled,shape=SEbox,inner sep=2pt,minimum height=6mm,fill=white,minimum width=12mm]
\tikzstyle{large dmap dag}=[draw,doubled,shape=SEbox,inner sep=2pt,minimum height=6mm,fill=white,minimum width=12mm]
\tikzstyle{very large dmap}=[draw,doubled,shape=NEbox,inner sep=2pt,minimum height=6mm,fill=white,minimum width=19.5mm]

\tikzstyle{muxbox}=[draw,shape=rectangle,minimum height=3mm,minimum width=3mm,fill=white]
\tikzstyle{dmuxbox}=[muxbox,doubled]

\tikzstyle{box}=[draw,shape=rectangle,inner sep=2pt,minimum height=6mm,minimum width=6mm,fill=white]
\tikzstyle{dbox}=[draw,doubled,shape=rectangle,inner sep=2pt,minimum height=6mm,minimum width=6mm,fill=white]
\tikzstyle{dmap}=[draw,doubled,shape=NEbox,inner sep=2pt,minimum height=6mm,fill=white]
\tikzstyle{dmapdag}=[draw,doubled,shape=SEbox,inner sep=2pt,minimum height=6mm,fill=white]
\tikzstyle{dmapadj}=[draw,doubled,shape=SEbox,inner sep=2pt,minimum height=6mm,fill=white]
\tikzstyle{dmaptrans}=[draw,doubled,shape=SWbox,inner sep=2pt,minimum height=6mm,fill=white]
\tikzstyle{dmapconj}=[draw,doubled,shape=NWbox,inner sep=2pt,minimum height=6mm,fill=white]

\tikzstyle{ddmap}=[draw,doubled,dashed,shape=NEbox,inner sep=2pt,minimum height=6mm,fill=white]
\tikzstyle{ddmapdag}=[draw,doubled,dashed,shape=SEbox,inner sep=2pt,minimum height=6mm,fill=white]
\tikzstyle{ddmapadj}=[draw,doubled,dashed,shape=SEbox,inner sep=2pt,minimum height=6mm,fill=white]
\tikzstyle{ddmaptrans}=[draw,doubled,dashed,shape=SWbox,inner sep=2pt,minimum height=6mm,fill=white]
\tikzstyle{ddmapconj}=[draw,doubled,dashed,shape=NWbox,inner sep=2pt,minimum height=6mm,fill=white]

\boxshape{sNEbox}{0pt}{3pt}
\boxshape{sSEbox}{0pt}{-3pt}
\boxshape{sNWbox}{3pt}{0pt}
\boxshape{sSWbox}{-3pt}{0pt}
\tikzstyle{smap}=[draw,shape=sNEbox,fill=white]
\tikzstyle{smapdag}=[draw,shape=sSEbox,fill=white]
\tikzstyle{smapadj}=[draw,shape=sSEbox,fill=white]
\tikzstyle{smaptrans}=[draw,shape=sSWbox,fill=white]
\tikzstyle{smapconj}=[draw,shape=sNWbox,fill=white]

\tikzstyle{dsmap}=[draw,dashed,shape=sNEbox,fill=white]
\tikzstyle{dsmapdag}=[draw,dashed,shape=sSEbox,fill=white]
\tikzstyle{dsmaptrans}=[draw,dashed,shape=sSWbox,fill=white]
\tikzstyle{dsmapconj}=[draw,dashed,shape=sNWbox,fill=white]

\boxshape{mNEbox}{0pt}{10pt}
\boxshape{mSEbox}{0pt}{-10pt}
\boxshape{mNWbox}{10pt}{0pt}
\boxshape{mSWbox}{-10pt}{0pt}
\tikzstyle{mmap}=[draw,shape=mNEbox]
\tikzstyle{mmapdag}=[draw,shape=mSEbox]
\tikzstyle{mmaptrans}=[draw,shape=mSWbox]
\tikzstyle{mmapconj}=[draw,shape=mNWbox]

\tikzstyle{mmapgray}=[draw,fill=gray!40!white,shape=mNEbox]
\tikzstyle{smapgray}=[draw,fill=gray!40!white,shape=sNEbox]

\makeatletter

\pgfdeclareshape{cornerpoint}{
\inheritsavedanchors[from=rectangle] 
\inheritanchorborder[from=rectangle]
\inheritanchor[from=rectangle]{center}
\inheritanchor[from=rectangle]{north}
\inheritanchor[from=rectangle]{south}
\inheritanchor[from=rectangle]{west}
\inheritanchor[from=rectangle]{east}
\backgroundpath{
\southwest \pgf@xa=\pgf@x \pgf@ya=\pgf@y
\northeast \pgf@xb=\pgf@x \pgf@yb=\pgf@y

\pgfmathsetmacro{\pgf@shorten@left}{\pgfkeysvalueof{/tikz/shorten left}}
\pgfmathsetmacro{\pgf@shorten@right}{\pgfkeysvalueof{/tikz/shorten right}}

\pgfpathmoveto{\pgfpoint{0.5 * (\pgf@xa + \pgf@xb)}{\pgf@ya - 5pt}}
\pgfpathlineto{\pgfpoint{\pgf@xa - 8pt + \pgf@shorten@left}{\pgf@yb - 1.5 * \pgf@shorten@left}}
\pgfpathlineto{\pgfpoint{\pgf@xa - 8pt + \pgf@shorten@left}{\pgf@yb}}
\pgfpathlineto{\pgfpoint{\pgf@xb + 8pt - \pgf@shorten@right}{\pgf@yb}}
\pgfpathlineto{\pgfpoint{\pgf@xb + 8pt - \pgf@shorten@right}{\pgf@yb - 1.5 * \pgf@shorten@right}}
\pgfpathclose
}
}

\pgfdeclareshape{cornercopoint}{
\inheritsavedanchors[from=rectangle] 
\inheritanchorborder[from=rectangle]
\inheritanchor[from=rectangle]{center}
\inheritanchor[from=rectangle]{north}
\inheritanchor[from=rectangle]{south}
\inheritanchor[from=rectangle]{west}
\inheritanchor[from=rectangle]{east}
\backgroundpath{
\southwest \pgf@xa=\pgf@x \pgf@ya=\pgf@y
\northeast \pgf@xb=\pgf@x \pgf@yb=\pgf@y

\pgfmathsetmacro{\pgf@shorten@left}{\pgfkeysvalueof{/tikz/shorten left}}
\pgfmathsetmacro{\pgf@shorten@right}{\pgfkeysvalueof{/tikz/shorten right}}

\pgfpathmoveto{\pgfpoint{0.5 * (\pgf@xa + \pgf@xb)}{\pgf@yb + 5pt}}
\pgfpathlineto{\pgfpoint{\pgf@xa - 8pt + \pgf@shorten@left}{\pgf@ya + 1.5 * \pgf@shorten@left}}
\pgfpathlineto{\pgfpoint{\pgf@xa - 8pt + \pgf@shorten@left}{\pgf@ya}}
\pgfpathlineto{\pgfpoint{\pgf@xb + 8pt - \pgf@shorten@right}{\pgf@ya}}
\pgfpathlineto{\pgfpoint{\pgf@xb + 8pt - \pgf@shorten@right}{\pgf@ya + 1.5 * \pgf@shorten@right}}
\pgfpathclose
}
}

\makeatother

\pgfkeyssetvalue{/tikz/shorten left}{0pt}
\pgfkeyssetvalue{/tikz/shorten right}{0pt}

\tikzstyle{kpoint common}=[draw,fill=white,inner sep=1pt,minimum height=4mm]
\tikzstyle{kpoint sc}=[shape=cornerpoint,kpoint common]
\tikzstyle{kpoint adjoint sc}=[shape=cornercopoint,kpoint common]
\tikzstyle{kpoint}=[shape=cornerpoint,shorten left=5pt,kpoint common]
\tikzstyle{kpoint adjoint}=[shape=cornercopoint,shorten left=5pt,kpoint common]
\tikzstyle{kpoint conjugate}=[shape=cornerpoint,shorten right=5pt,kpoint common]
\tikzstyle{kpoint transpose}=[shape=cornercopoint,shorten right=5pt,kpoint common]
\tikzstyle{kpoint symm}=[shape=cornerpoint,shorten left=5pt,shorten right=5pt,kpoint common]

\tikzstyle{black kpoint}=[shape=cornerpoint,shorten left=5pt,kpoint common,fill=black,font=\color{white}]
\tikzstyle{black kpoint adjoint}=[shape=cornercopoint,shorten left=5pt,kpoint common,fill=black,font=\color{white}]
\tikzstyle{black kpointadj}=[shape=cornercopoint,shorten left=5pt,kpoint common,fill=black,font=\color{white}]

\tikzstyle{black dkpoint}=[shape=cornerpoint,shorten left=5pt,kpoint common,fill=black, doubled,font=\color{white}]
\tikzstyle{black dkpoint adjoint}=[shape=cornercopoint,shorten left=5pt,kpoint common,fill=black, doubled,font=\color{white}]
\tikzstyle{black dkpointadj}=[shape=cornercopoint,shorten left=5pt,kpoint common,fill=black, doubled,font=\color{white}] 

\tikzstyle{kpointdag}=[kpoint adjoint]
\tikzstyle{kpointadj}=[kpoint adjoint]
\tikzstyle{kpointconj}=[kpoint conjugate]
\tikzstyle{kpointtrans}=[kpoint transpose]

\tikzstyle{big kpoint}=[kpoint, minimum width=1.2 cm, minimum height=8mm, inner sep=4pt, text depth=3mm]

\tikzstyle{wide kpoint}=[kpoint, minimum width=1 cm, inner sep=2pt]
\tikzstyle{wide kpointdag}=[kpointdag, minimum width=1 cm, inner sep=2pt]
\tikzstyle{wide kpointconj}=[kpointconj, minimum width=1 cm, inner sep=2pt]
\tikzstyle{wide kpointtrans}=[kpointtrans, minimum width=1 cm, inner sep=2pt]

\tikzstyle{gray kpoint}=[kpoint,fill=gray!50!white]
\tikzstyle{gray kpointdag}=[kpointdag,fill=gray!50!white]
\tikzstyle{gray kpointadj}=[kpointadj,fill=gray!50!white]
\tikzstyle{gray kpointconj}=[kpointconj,fill=gray!50!white]
\tikzstyle{gray kpointtrans}=[kpointtrans,fill=gray!50!white]

\tikzstyle{gray dkpoint}=[kpoint,fill=gray!50!white,doubled]
\tikzstyle{gray dkpointdag}=[kpointdag,fill=gray!50!white,doubled]
\tikzstyle{gray dkpointadj}=[kpointadj,fill=gray!50!white,doubled]
\tikzstyle{gray dkpointconj}=[kpointconj,fill=gray!50!white,doubled]
\tikzstyle{gray dkpointtrans}=[kpointtrans,fill=gray!50!white,doubled]

\tikzstyle{white label}=[draw,fill=white,rectangle,inner sep=0.7 mm]
\tikzstyle{gray label}=[draw,fill=gray!50!white,rectangle,inner sep=0.7 mm]
\tikzstyle{black label}=[draw,fill=black,rectangle,inner sep=0.7 mm]

\tikzstyle{dkpoint}=[kpoint,doubled]
\tikzstyle{wide dkpoint}=[wide kpoint,doubled]
\tikzstyle{dkpointdag}=[kpoint adjoint,doubled]
\tikzstyle{wide dkpointdag}=[wide kpointdag,doubled]
\tikzstyle{dkcopoint}=[kpoint adjoint,doubled]
\tikzstyle{dkpointadj}=[kpoint adjoint,doubled]
\tikzstyle{dkpointconj}=[kpoint conjugate,doubled]
\tikzstyle{dkpointtrans}=[kpoint transpose,doubled]

\tikzstyle{kscalar}=[kpoint common, shape=EBox, inner xsep=-1pt, inner ysep=3pt,font=\small]
\tikzstyle{kscalarconj}=[kpoint common, shape=WBox, inner xsep=-1pt, inner ysep=3pt,font=\small]

\tikzstyle{spekpoint}=[kpoint sc,minimum height=5mm,inner sep=3pt]
\tikzstyle{spekcopoint}=[kpoint adjoint sc,minimum height=5mm,inner sep=3pt]

\tikzstyle{dspekpoint}=[spekpoint,doubled]
\tikzstyle{dspekcopoint}=[spekcopoint,doubled]

 \tikzstyle{upground}=[circuit ee IEC,thick,ground,rotate=90,scale=1.0]
 \tikzstyle{downground}=[circuit ee IEC,thick,ground,rotate=-90,scale=1.0]
 \tikzstyle{bigground}=[regular polygon,regular polygon sides=3,draw=gray,scale=0.50,inner sep=-0.5pt,minimum width=10mm,fill=gray]

\tikzstyle{arrs}=[-latex,font=\small,auto]
\tikzstyle{arrow plain}=[arrs]
\tikzstyle{arrow dashed}=[dashed,arrs]
\tikzstyle{arrow bold}=[very thick,arrs]
\tikzstyle{arrow hide}=[draw=white!0,-]
\tikzstyle{arrow reverse}=[latex-]
\tikzstyle{cdnode}=[]

\newcommand{\unit}{\dotunit{dot}}

\input{zx.tikzdefs}

\tikzstyle{box}=[shape=rectangle, text height=1.5ex, text depth=0.25ex, yshift=0.5mm, fill=white, draw=black, minimum height=5mm, yshift=-0.5mm, minimum width=5mm, font={\small}]
\tikzstyle{Z dot}=[inner sep=0mm, minimum size=2mm, shape=circle, draw=black, fill={rgb,255: red,221; green,255; blue,221}]
\tikzstyle{Z phase dot}=[minimum size=5mm, font={\footnotesize\boldmath}, shape=rectangle, rounded corners=2mm, inner sep=0.2mm, outer sep=-2mm, scale=0.8, tikzit shape=circle, draw=black, fill={rgb,255: red,221; green,255; blue,221}, tikzit draw=blue]
\tikzstyle{Z dot bold}=[Z dot, thick]
\tikzstyle{Z phase dot bold}=[Z phase dot, thick]
\tikzstyle{X dot}=[Z dot, shape=circle, draw=black, fill={rgb,255: red,255; green,136; blue,136}]
\tikzstyle{X phase dot}=[Z phase dot, tikzit shape=circle, tikzit draw=blue, fill={rgb,255: red,255; green,136; blue,136}, font={\footnotesize\boldmath}]
\tikzstyle{X dot bold}=[X dot, thick]
\tikzstyle{X phase dot bold}=[X phase dot, thick]
\tikzstyle{hadamard}=[fill=yellow, draw=black, shape=rectangle, inner sep=0.6mm, minimum height=1.5mm, minimum width=1.5mm]
\tikzstyle{hadamard bold}=[hadamard, thick]
\tikzstyle{vertex}=[inner sep=0mm, minimum size=1mm, shape=circle, draw=black, fill=black]
\tikzstyle{vertex set}=[inner sep=0mm, minimum size=1mm, shape=circle, draw=black, fill=white, font={\footnotesize\boldmath}]
\tikzstyle{ground}=[ground, anchor=center, child anchor=center, text='$\ground$']
\tikzstyle{ground bold}=[ground, thick]

\tikzstyle{rfarr}=[draw, signal, fill=black, signal to=east, signal from=west, inner sep=1pt, minimum height=6pt]
\tikzstyle{lfarr}=[draw, signal, fill=black, signal to=west, signal from=east, inner sep=1pt, minimum height=6pt]
\tikzstyle{ufarr}=[draw, signal, fill=black, signal to=north, signal from=south, inner sep=1pt, minimum width=6pt]
\tikzstyle{dfarr}=[draw, signal, fill=black, signal to=south, signal from=north, inner sep=1pt, minimum width=6pt]

\tikzstyle{dashed edge}=[-, dashed, thick]
\tikzstyle{hadamard edge}=[-, dashed, dash pattern=on 2pt off 0.5pt, thick, draw={rgb,255: red,68; green,136; blue,255}]
\tikzstyle{brace edge}=[-, tikzit draw=blue, decorate, decoration={brace,amplitude=1mm,raise=-1mm}]
\tikzstyle{bold edge}=[-, thick]
\tikzstyle{diredge}=[->]

\tikzstyle{border edge}=[-, dashed, dash pattern=on 2pt off 0.5pt, thick, draw={rgb,255: red,255; green,13; blue,20}]

\tikzstyle{bit clone}=[small black dot]

\tikzstyle{qubit edge}=[-]
\tikzstyle{bit edge}=[-, double distance=1.5pt]

\tikzstyle{gather}=[fill=zx_grey, draw=black, tikzit category=scal, rounded corners=0.8mm, regular polygon, regular polygon sides=3, shape border rotate=-90, inner sep=1.6pt]
\tikzstyle{split}=[fill=zx_grey, draw=black, tikzit category=scal, rounded corners=0.8mm, regular polygon, regular polygon sides=3, shape border rotate=90, inner sep=1.6pt]

\tikzstyle{derelict}=[fill=zx_black, draw=black, tikzit category=scal, regular polygon, regular polygon sides=3, shape border rotate=90, inner sep=1.6pt]
\tikzstyle{clone}=[inner sep=0mm, minimum size=2mm, shape=circle, draw=black, fill={rgb,255: red,255; green,255; blue,255}]

\tikzstyle{bang box}=[shape=rectangle, text height=1.5ex, text depth=0.25ex, yshift=0.5mm, draw=black, minimum height=10mm, yshift=-0.5mm, minimum width=10mm, font={\small},
dashed]

\graphicspath{figures/}

\newcommand{\bigo}{\ensuremath{\mathcal{O}}}
\newcommand{\N}{\ensuremath{\mathbb{N}}}
\newcommand{\R}{\ensuremath{\mathbb{R}}}

\newcommand{\Ftwo}{\ensuremath{\mathbb{F}_2}}

\newcommand{\diag}{\ensuremath{\mathscr{D}}} 
\ifdefined\C
  \renewcommand{\C}{\ensuremath{\mathbb{C}}}
\else
  \newcommand{\C}{\ensuremath{\mathbb{C}}}
\fi
\ifdefined\U
  \renewcommand{\U}{\mathtt{U}}
\else
  \newcommand{\U}{\mathtt{U}}
\fi
\ifdefined\R
  \renewcommand{\R}{\mathtt{R}}
\else
  \newcommand{\R}{\mathtt{R}}
\fi
\newcommand{\ketbra}[2]{\ket{#1}\!\!\bra{#2}}

\newcommand{\imod}[2]{{#1}\ \text{mod}\ {#2}} 
\newcommand{\idiv}[2]{{#1}\ \text{div}\ {#2}} 

\newcommand{\meas}{\mathtt{meas}}
\newcommand{\new}{\mathtt{new}}
\newcommand{\bit}{\mathtt{B}}
\newcommand{\qubit}{\mathtt{Q}}
\renewcommand{\unit}{\mathtt{Unit}}
\newcommand{\nat}{\mathtt{Nat}}
\renewcommand{\vec}[2]{\mathtt{Vec}\ {#2}\ {#1}}
\newcommand{\Qlet}[3]{\mathtt{let\ }\allowbreak #1 \allowbreak=\allowbreak #2 \mathtt{\ in\ } \allowbreak #3}
\newcommand{\Qfold}{\mathtt{fold}}
\newcommand{\Qmap}{\mathtt{map}}
\newcommand{\Qaccumap}{\mathtt{accuMap}}
\newcommand{\Qsplit}{\mathtt{split}}
\newcommand{\Qappend}{\mathtt{append}}
\newcommand{\Qcompose}{\mathtt{compose}}
\newcommand{\Qfor}[3]{\mathtt{for\ }{#1}\mathtt{\ in\ }{#2}\mathtt{\ do\ }{#3}}
\newcommand{\ifz}[3]{\mathtt{ifz\ }{#1}\mathtt{\ then\ }{#2}\mathtt{\ else\ }{#3}}
\newcommand{\vnil}{\mathtt{VNil}}
\newcommand{\Qrange}{\mathtt{range\ }}
\newcommand{\lambdaD}{\ensuremath{\lambda_D}}
\newcommand{\Qdrop}{\mathtt{drop}}

\newcommand{\trans}[1]{\left\llbracket{#1}\right\rrbracket}
\newcommand{\eval}[1]{\left\lfloor{#1}\right\rfloor}

\newcommand{\interpret}[1]{\left\{\!\!\!\left\{{#1}\right\}\!\!\!\right\}}

\usetikzlibrary{circuits.ee.IEC}
\usetikzlibrary{shapes.gates.logic.US,shapes.gates.logic.IEC}

\newcommand{\ground}{%
    \begin{tikzpicture}[circuit ee IEC,yscale=1.0,xscale=1.0]
    \draw[solid,arrows=-] (0,1ex) to (0,0) node[anchor=center,ground,rotate=-90,xshift=.66ex] {};
    \end{tikzpicture}
}

\newcommand{\zxGND}{\ensuremath{\text{ZX}_{\ground}}}
\newcommand{\szxGND}{\ensuremath{\text{SZX}_{\ground}}}

\definecolor{orange}{RGB}{200,127,0}

\numberwithin{equation}{section}
\theoremclass{Theorem}
\theoremstyle{plain}
\theorembodyfont{\upshape}

\theoremstyle{plain}
\theorembodyfont{\upshape}
\newtheorem{lemma}[equation]{Lemma}
\newtheorem*{lemmaN}[equation]{Lemma} 

\theoremstyle{plain}
\theorembodyfont{\upshape}
\theoremsymbol{\ensuremath{\square}}
\newtheorem*{proof}{Proof}

\theoremstyle{plain}
\theorembodyfont{\upshape}
\theoremsymbol{}

\theoremstyle{plain}
\theoremsymbol{}

\newif\ifreduced
\newcommand{\reducedorlong}[2]{\ifreduced{#1}\else{#2}\fi}

\title{Encoding High-level Quantum Programs as SZX-diagrams}
\author{
    Augustin Borgna
    \institute{CNRS LORIA, Inria-MOCQUA, \\ Université de Lorraine}
    \institute{CNRS, LMF, \\ Université Paris-Saclay}
    \email{agustin.borgna@loria.fr}
    \and
    Rafael Romero
    \institute{CONICET, Instituto de Ciencias de la Computación \\ Universidad de Buenos Aires, Buenos Aires, Argentina}
    \institute{PEDECIBA,\\ Universidad de la República-MEC. Montevideo, Uruguay}
    \email{lromero@dc.uba.ar}
}
\def\titlerunning{Encoding High-level Quantum Programs as SZX-diagrams}
\def\authorrunning{A. Borgna, R. Romero}

\begin{document}

\maketitle

\begin{abstract}
The Scalable ZX-calculus is a compact graphical language used to
reason about linear maps between quantum states. These diagrams have multiple
applications, but they frequently have to be constructed in a case-by-case basis.
In this work we present a method to encode quantum programs implemented
in a fragment of the linear dependently typed Proto-Quipper-D language as families of SZX-diagrams.
We define a subset of translatable Proto-Quipper-D programs
and show that our procedure is able to encode non-trivial algorithms
as diagrams that grow linearly on the size of the program.
\end{abstract}

\section{Introduction}%
\label{sec:introduction}

The ZX calculus~\cite{vdw_working_cs_zx} has been used as intermediary representation language
for quantum programs in optimization methods~\cite{DKPW2019qcircSimpl,borgna_hybrid_2021,Backens_2021}
and in the design of error correcting schemes~\cite{de_beaudrap_zx_2020}.
The highly flexible representation of linear maps as open graphs
with a complete formal rewriting system
and the multiple extensions adapted to represent
different sets of quantum primitives
have proven useful in reasoning about the properties of quantum circuits.

Quantum operations are usually represented as quantum circuits
composed by primitive gates operating over a fixed number of qubits.
The ZX calculus has a close correspondence to this model and is similarly 
limited to representing operations at a single-qubit level.
In this work we will focus on the Scalable ZX extension~\cite{carette_szx-calculus_2019},
which generalizes the ZX diagrams to work with arbitrary qubit registers
using a compact representation.
Previous work~\cite{carette_quantum_2021} has shown that the SZX calculus
is capable of encoding nontrivial algorithms
via the presentation of multiple hand-written examples.
For an efficient usage as an intermediate representation language,
we require an automated compilation method from quantum programming languages to SZX diagrams.
While ZX diagrams can be directly obtained from a program compiled to a quantum circuit,
to the best of our knowledge there is no efficient method leveraging the parametricity of the SZX calculus.

There exist several quantum programming languages capable of encoding
high-level parametric programs~\cite{qiskit,cirq,Steiger2018projectQ}.
Quipper~\cite{Green2013quipper} is a language for quantum computing
capable or generating families of quantum operations indexed by parameters.
These parameters need to be instantiated at compile time to generate
concrete quantum circuit representations.
Quipper has multiple formal specifications, in this work we focus on the
linear dependently typed Proto-Quipper-D formalization~\cite{fu_linear_2021,fu_tutorial_2020}
to express high-level programs with integer parameters.

The contributions of this article the following.
We introduce a \textit{list initialization} notation to represent multiple elements
of a SZX diagram family composed in parallel.
We formally define a fragment of Proto-Quipper-D programs that can be described as families of diagrams.
Then we present a novel compilation method that encodes quantum programs as families of SZX diagrams
and demonstrate the codification and translation of a nontrivial algorithm using our procedure.

In Section~\ref{sec:background} we outline both languages and introduce the list
initialization notation. In Section~\ref{sec:fragment} we define the restricted
Proto-Quipper-D fragment. In Section~\ref{sec:translation} we introduce the
translation into SZX diagrams. Finally, in Section~\ref{sec:qft-example} we
demonstrate an encoding of the Quantum Fourier Transform algorithm using our
method. The proofs of the lemmas stated in this work can be found
\reducedorlong{in the long version of this document published in ArXiV}{in
Appendix~\ref{sec:proofs}}.

\section{Background}%
\label{sec:background}

We describe a quantum state as a system of $n$ qubits
corresponding to a vector in the $\C^{2^n}$ Hilbert space.
We may partition the set of qubits into multi-qubit \textit{registers}
representing logically related subsets.
Quantum computations under the QRAM model correspond to compositions of unitary
operators between these quantum states, called quantum gates.
Additionally, the qubits may be initialized on a set state and measured.

High-level programs can be encoded in Quipper~\cite{Green2013quipper},
a Haskell-like programming language for describing quantum computations.
In this work we use a formalization of the language called
Proto-Quipper-D\cite{fu_tutorial_2020} with support for linear and dependent types.
Concrete quantum operations correspond to linear functions between quantum states,
generated as a composition of primitive operations that can be described directly
as a quantum circuit.
Generic circuits may have additional parameters that must fixed at compilation time
to produce the corresponding quantum circuit.

In Section~\ref{sec:fragment} we describe a restricted fragment of the Proto-Quipper-D
language containing the relevant operations for the work presented in this paper.

\subsection{The Scalable ZX-calculus}%
\label{sec:szx}

The ZX calculus~\cite{vdw_working_cs_zx} is a formal graphical language
that encodes linear maps between quantum states.
Multiple extensions to the calculus have been proposed.
We first present the base calculus with the grounded-ZX extension,
denoted \zxGND~\cite{ground},
to allow us to encode quantum state measurement operations.
A \zxGND\ diagram is generated by the following primitives,
in addition to parallel and serial composition:
\[
    \input{szx/elem/spiderZ.tikz} : n_1 \to m_1
    \qquad
    \input{szx/elem/spiderX.tikz} : n_1 \to m_1
    \qquad
    \begin{tikzpicture}
	\begin{pgfonlayer}{nodelayer}
		\node [style=hadamard] (0) at (0, 0) {};
		\node [style=none] (1) at (1, 0) {};
		\node [style=none] (2) at (-1, 0) {};
	\end{pgfonlayer}
	\begin{pgfonlayer}{edgelayer}
		\draw [in=180, out=0, looseness=0.75] (0) to (1.center);
		\draw [in=180, out=0, looseness=0.75] (2.center) to (0);
	\end{pgfonlayer}
\end{tikzpicture}
 : 1_1 \to 1_1
    \qquad
    \begin{tikzpicture}
	\begin{pgfonlayer}{nodelayer}
		\node [style=none] (2) at (-1, 0) {};
		\node [style=ground, rotate=90] (4) at (0, 0) {$\ground$};
	\end{pgfonlayer}
	\begin{pgfonlayer}{edgelayer}
		\draw (2.center) to (4.center);
	\end{pgfonlayer}
\end{tikzpicture}
 : 1_1 \to 0_1
\]
\[
    \begin{tikzpicture}
	\begin{pgfonlayer}{nodelayer}
		\node [style=none] (1) at (2, 0) {};
		\node [style=none] (2) at (0, 0) {};
	\end{pgfonlayer}
	\begin{pgfonlayer}{edgelayer}
		\draw [in=180, out=0] (2.center) to (1.center);
	\end{pgfonlayer}
\end{tikzpicture}
 : 1_1 \to 1_1
    \qquad
    \begin{tikzpicture}
	\begin{pgfonlayer}{nodelayer}
		\node [style=none] (1) at (0, 0.5) {};
		\node [style=none] (2) at (0, -0.5) {};
	\end{pgfonlayer}
	\begin{pgfonlayer}{edgelayer}
		\draw [in=180, out=180, looseness=1.75] (2.center) to (1.center);
	\end{pgfonlayer}
\end{tikzpicture}
 : 0_1 \to 2_1
    \qquad
    \begin{tikzpicture}
	\begin{pgfonlayer}{nodelayer}
		\node [style=none] (1) at (0, 0.5) {};
		\node [style=none] (2) at (0, -0.5) {};
	\end{pgfonlayer}
	\begin{pgfonlayer}{edgelayer}
		\draw [in=0, out=0, looseness=1.75] (2.center) to (1.center);
	\end{pgfonlayer}
\end{tikzpicture}
 : 2_1 \to 0_1
    \qquad
    \input{szx/elem/swap.tikz} : 2_1 \to 2_1
    \qquad
    \begin{tikzpicture}
	\begin{pgfonlayer}{nodelayer}
		\node [style=bang box, minimum width=1em, minimum height=1em] (0) at (0, 0) {};
	\end{pgfonlayer}
\end{tikzpicture}
 : 0_1 \to 0_1
\]
where $n_k$ represents the n-tensor of k-qubit registers,
the green and red nodes are called Z and X spiders,
$\alpha \in [0,2\pi)$ is the phase of the spiders,
and the yellow square is called the Hadamard node.
These primitives allow us to encode any quantum operation,
but they can become impractical when working with multiple qubit registers.

The SZX calculus~\cite{carette_szx-calculus_2019,carette_quantum_2021}
is a \textit{Scalable} extension to the ZX-calculus
that generalizes the primitives to work with arbitrarily sized qubit registers.
This facilitates the representation of diagrams with repeated structure
in a compact manner. 
Carette et al.~\cite{carette_quantum_2021} show that the scalable and grounded extensions
can be directly composed.
Will refer to the resulting $\szxGND$-calculus as SZX for simplicity.
Bold wires in a SZX diagram are tagged with a non-negative integer representing
the size of the qubit register they carry, and other generators are marked in bold 
to represent a parallel application over each qubit in the register.
Bold spiders with multiplicity $k$ are tagged with $k$-sized vectors of phases
$\overline\alpha = \alpha_1 :: \dots :: \alpha_k$.
The natural extension of the ZX generators correspond to the following primitives:
\[
    \input{szx/elem/spiderZ-szx.tikz} : n_k \to m_k
    \qquad
    \input{szx/elem/spiderX-szx.tikz} : n_k \to m_k
    \qquad
    \input{szx/elem/hadamard-szx.tikz} : 1_k \to 1_k
    \qquad
    \begin{tikzpicture}
	\begin{pgfonlayer}{nodelayer}
		\node [style=none] (2) at (-1, 0) {};
		\node [style=none] (3) at (-0.5, 0.25) {$\scriptstyle{k}$};
		\node [style=ground bold, rotate=90] (4) at (0, 0) {$\ground$};
	\end{pgfonlayer}
	\begin{pgfonlayer}{edgelayer}
		\draw [style=bold edge] (2.center) to (4.center);
	\end{pgfonlayer}
\end{tikzpicture}
 : 1_k \to 0_0
\]
\[
    \begin{tikzpicture}
	\begin{pgfonlayer}{nodelayer}
		\node [style=none] (1) at (2, 0) {};
		\node [style=none] (2) at (0, 0) {};
		\node [style=none] (3) at (1, 0.25) {$\scriptstyle{k}$};
	\end{pgfonlayer}
	\begin{pgfonlayer}{edgelayer}
		\draw [style=bold edge, in=180, out=0] (2.center) to (1.center);
	\end{pgfonlayer}
\end{tikzpicture}
 : 1_k \to 1_k
    \qquad
    \begin{tikzpicture}
	\begin{pgfonlayer}{nodelayer}
		\node [style=none] (1) at (0, 0.5) {};
		\node [style=none] (2) at (0, -0.5) {};
		\node [style=none] (3) at (-0.75, 0) {$\scriptstyle{k}$};
	\end{pgfonlayer}
	\begin{pgfonlayer}{edgelayer}
		\draw [style=bold edge, in=180, out=180, looseness=1.75] (2.center) to (1.center);
	\end{pgfonlayer}
\end{tikzpicture}
 : 0_0 \to 2_k
    \qquad
    \begin{tikzpicture}
	\begin{pgfonlayer}{nodelayer}
		\node [style=none] (1) at (0, 0.5) {};
		\node [style=none] (2) at (0, -0.5) {};
		\node [style=none] (3) at (0.75, 0) {$\scriptstyle{k}$};
	\end{pgfonlayer}
	\begin{pgfonlayer}{edgelayer}
		\draw [style=bold edge, in=0, out=0, looseness=1.75] (2.center) to (1.center);
	\end{pgfonlayer}
\end{tikzpicture}
 : 2_k \to 0_0
    \qquad
    \input{szx/elem/swap-szx.tikz} : 1_k \otimes 1_l \to 1_l \otimes 1_k
    \qquad
    \begin{tikzpicture}
	\begin{pgfonlayer}{nodelayer}
		\node [style=bang box, minimum width=1em, minimum height=1em] (0) at (0, 0) {};
	\end{pgfonlayer}
\end{tikzpicture}
 : 0_k \to 0_k
\]
Wires of multiplicity zero are equivalent to the empty mapping.
We may omit writing the wire multiplicity if it can be deduced by context. 

The extension defines two additional generators; a \textit{split} node to split
registers into multiple wires, and a function arrow to apply arbitrary functions
over a register. In this work we restrict the arrow functions to permutations
$\sigma : [0 \dots k) \to [0 \dots k)$ that rearrange the order of the wires.
Using the split node and the wire primitives can derive the rotated version,
which we call a \textit{gather}.
\[
    \input{szx/elem/split.tikz} : 1_{n+m} \to 1_n \otimes 1_m
    \qquad
    \input{szx/elem/gather.tikz} : 1_n \otimes 1_m \to 1_{n+m}
    \qquad
    \begin{tikzpicture}
	\begin{pgfonlayer}{nodelayer}
		\node [style=none] (1) at (1, 0) {};
		\node [style=none] (2) at (-1, 0) {};
		\node [style=none] (3) at (0, 0.5) {$\sigma$};
		\node [style=rfarr] (4) at (0, 0) {};
	\end{pgfonlayer}
	\begin{pgfonlayer}{edgelayer}
		\draw [style=bold edge] (2.center) to (4);
		\draw [style=bold edge] (4) to (1.center);
	\end{pgfonlayer}
\end{tikzpicture}
 : 1_k \to 1_k
\]

The rewriting rules of the calculus imply that a SZX diagrams can be considered
as an open graph where only the topology of its nodes and edges matters.
In the translation process we will make repeated use of the following
reductions rules to simplify the diagrams:
\[
    \input{szx/rules/split-gather.tikz}
    \ \stackrel{\mathbf{(sg)}}{=}\ %
    \begin{tikzpicture}
	\begin{pgfonlayer}{nodelayer}
		\node [style=none] (1) at (2, 0) {};
		\node [style=none] (2) at (0, 0) {};
		\node [style=none] (3) at (1, 0.25) {$\scriptstyle{n+m}$};
	\end{pgfonlayer}
	\begin{pgfonlayer}{edgelayer}
		\draw [style=bold edge, in=180, out=0] (2.center) to (1.center);
	\end{pgfonlayer}
\end{tikzpicture}

    \qquad\qquad
    \input{szx/rules/gather-split.tikz}
    \ \stackrel{\mathbf{(gs)}}{=}\ %
    \input{szx/rules/gather-split-2.tikz}
\]
We may also depict composition of gathers as single multi-legged generators.
In an analogous manner, we will use a legless gather $\begin{tikzpicture}
	\begin{pgfonlayer}{nodelayer}
		\node [style=split] (0) at (0, 0) {};
		\node [style=none] (1) at (-0.75, 0) {};
	\end{pgfonlayer}
	\begin{pgfonlayer}{edgelayer}
		\draw [style=bold edge] (1.center) to (0.center);
	\end{pgfonlayer}
\end{tikzpicture}
$ to terminate wires with cardinality zero.
This could be encoded as the zero-multiplicity spider $\begin{tikzpicture}
	\begin{pgfonlayer}{nodelayer}
		\node [style=none] (1) at (-0.75, 0) {};
		\node [style=Z phase dot bold] (2) at (0, 0) {$\scriptstyle{[\,]}$};
	\end{pgfonlayer}
	\begin{pgfonlayer}{edgelayer}
		\draw [style=bold edge] (1.center) to (2);
	\end{pgfonlayer}
\end{tikzpicture}
$, which represents the empty mapping.

\reducedorlong{}{Refer to Appendix~\ref{sec:szx-extended} for a complete definition
of the rewriting rules and the interpretation of the SZX calculus.}
Cf.~\cite{carette_quantum_2021} for a description of the calculus including the generalized arrow generators.

Carette et al.~\cite{carette_quantum_2021} showed that the SZX calculus
can encode the repetition of a function $f : 1_n \to 1_n$ an arbitrary number of times $k\geq 1$ as follows:
\[
    \input{szx/repeat.tikz}
    \quad=\quad%
    \left(\input{szx/repeat-2.tikz}\right)^k
\]
where $f^k$ corresponds to $k$ parallel applications of $f$.
With a simple modification this construction can be used to encode an accumulating map operation.
\begin{lemma}%
    \label{lem:accumap}
    Let $g : 1_n \otimes 1_s \to 1_{m} \otimes 1_s$ and $k\geq 1$, then
\[
    \input{szx/accumap.tikz}
    \quad=\quad%
    \input{szx/accumap-2.tikz}
\]
\end{lemma}
As an example, given a list $N = [n_1, n_2, n_3]$ and a starting accumulator value $x_0$,
this construction would produce the mapping $([n_1, n_2, n_3], x_0) \mapsto ([m_1, m_2, m_3], x_3)$ where
$(m_i, x_i) = g(n_i, x_{i-1})$ for $i \in [1,3]$.


\subsection{SZX diagram families and list instantiation}

We introduce the definition of a family of SZX diagrams $D: \N^k \to \diag$
as a function from $k$ integer \textit{parameters} to SZX diagrams.
We require the structure of the diagrams to be the same for all elements in the family,
parameters may only alter the wire tags and spider phases.
Partial application is allowed, we write $D(n)$ to fix the first parameter of $D$.


Since instantiations of a family share the same structure,
we can compose them in parallel by merging the different values of wire tags and spider phases.
We introduce a shorthand for instantiating a family of diagrams on multiple values
and combining the resulting diagrams in parallel.
This definition is strictly more general than the \textit{thickening endofunctor}
presented by Carette et al.~\cite{carette_quantum_2021},
which replicates a concrete diagram in parallel.
A \textit{list instantiation} of a family of diagrams $D: \N^{k+1} \to \diag$ over
a list $N$ of integers is written as $(D(n), n \in N)$.
This results in a family with one fewer parameter, $(D(n), n \in N): \N^k \to \diag$. 
We graphically depict a list instantiation as a dashed box in a diagram, as follows.
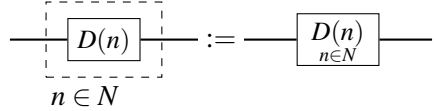
\[\input{szx/list-instantiation-box.tikz} := \input{szx/list-instantiation.tikz}\]

The definition of the list instantiation operator is given recursively
on the construction of $D$ in Figure~\ref{fig:list-instantiation}.
On the diagram wires we use $v(N)$ to denote the wire cardinality $\sum_{n \in N} v(n)$,
$\overrightarrow\alpha(N)$ for the concatenation of phase vectors
$\overrightarrow\alpha(n_1) :: \dots :: \overrightarrow\alpha(n_m)$,
and $\sigma(N)$ for the composition of permutations $\bigotimes_{n \in N} \sigma(n)$.
In general, a permutation arrow $\sigma(N,v,w)$ instantiated in concrete values
can be replaced by a reordering of wires between two gather gates%
\reducedorlong{}{ using the rewrite rule $\bf{(p)}$}.

\begin{figure}[tb]
\begin{mdframed}
    Given $D: \N^{k+1} \to \diag$, $N = [n_1, \dots, n_m] \in \N^m$,
    \[
        ((D_1 \otimes D_2)(n), n \in N) := (D_1(n), n \in N) \otimes (D_2(n), n \in N) 
        \qquad
        \input{szx/list-instantiation/wire.tikz}
        :=
        \input{szx/list-instantiation/wire-flat.tikz}
    \]
    \[
        ((D_2 \circ D_1)(n), n \in N) := (D_2(n), n \in N) \circ (D_1(n), n \in N) 
        \qquad
        \input{szx/list-instantiation/ground.tikz}
        :=
        \input{szx/list-instantiation/ground-flat.tikz}
    \]
    \[
        \input{szx/list-instantiation/hadam.tikz}
        :=
        \input{szx/list-instantiation/hadam-flat.tikz}
        \qquad
        \input{szx/list-instantiation/arrow.tikz}
        :=
        \input{szx/list-instantiation/arrow-flat.tikz}
    \]
    \[
        \input{szx/list-instantiation/spider-Z.tikz}
        :=
        \input{szx/list-instantiation/spider-Z-flat.tikz}
        \qquad
        \input{szx/list-instantiation/spider-X.tikz}
        :=
        \input{szx/list-instantiation/spider-X-flat.tikz}
    \]
    \[
        \input{szx/list-instantiation/gather.tikz}
        :=
        \input{szx/list-instantiation/gather-flat.tikz}
    \]
    Where $\sigma(N, v, w) \in \Ftwo^{v(N)+w(N) \times v(N)+w(N)}$ is the permutation defined as the matrix
    \[
        \sigma(N, v, w) = \begin{pmatrix}\sigma_f^N \vert \sigma_g^N \end{pmatrix}, \quad
        \sigma_f^{N} \in \Ftwo^{v(N)+w(N) \times v(N)},
        \ \sigma_g^{N} \in \Ftwo^{v(N)+w(N) \times w(N)}
    \]
    \[
    \begin{array}{llll}
        \sigma_f^{[\,]} = Id_0
        \quad &
        \sigma_f^{n::N'} =
            \begin{pmatrix}
                Id_{v(n)} & 0 \\
                0 & 0 \\
                0 & \sigma_f^{N'}
            \end{pmatrix}
        \quad &
        \sigma_g^{[\,]} = Id_0
        \quad &
        \sigma_g^{n::N'} =
            \begin{pmatrix}
                0 & 0 \\
                Id_{w(n)} & 0 \\
                0 & \sigma_g^{N'}
            \end{pmatrix}
    \end{array}
    \]
    \caption{Definition of the list instantiation operator.}%
    \label{fig:list-instantiation}
\end{mdframed}
\end{figure}

\begin{lemma}%
    \label{lem:list-instantiation}
    For any diagram family $D$, $n_0 : \N$, $N : \N^k$,
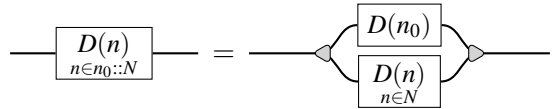
    \[
        \input{szx/list-instantiation-append.tikz}
        \;=\;
        \input{szx/list-instantiation-append-split.tikz}
    \]
\end{lemma}

\begin{lemma}%
    \label{lem:list-init-zero}
    A diagram family initialized with the empty list corresponds to the empty map.
    For any diagram family $D$,
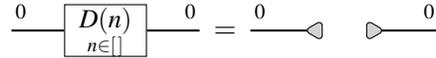
    \[
        \input{szx/list-instantiation-empty.tikz}
        \;=\;
        \input{szx/list-instantiation-empty-flat.tikz}
    \]
\end{lemma}

\begin{lemma}%
    \label{lem:list-instantiation-linear}
    The list instantiation procedure on an $n$-node diagram family adds
    $\bigo(n)$ nodes to the original diagram.
\end{lemma}

\section{The \texorpdfstring{$\lambdaD$}{lambda sub D} calculus}%
\label{sec:fragment}

We first define a base language from which to build our translation. In this section we present the calculus $\lambdaD$, as a subset of the strongly normalizing Proto-Quipper-D programs. Terms are inductively defined by:
\begin{align*}
    M, N, L :=\; & x \;|\; C \;|\; \R \;|\; \; \U \;|\; 0 \;|\; 1 \;|\; n \;|\; \meas \;|\; \new \;|
            \lambda x^S. M \;|\; M \ N \;|\; \lambda' x^P. M \;|\; M\ @\ N \;|\;\\
        & \star \;|\; M \otimes N \;|\; \Qlet{x^{S_1} \otimes y^{S_2}}{M}{N} \;|\;M;N \;|\; \\
        & \vnil^A \;|\; M :: N \;|\; \Qlet{x^S :: y^{\vec S n}}{M}{N} \\
        & M \square N \;|\; \ifz{L}{M}{N} \;|\; \Qfor{k^P}{M}{N}
\end{align*}

Where $C$ is a set of implicit bounded recursive primitives used for operating
with vectors and iterating functions. $n\in\mathbb{N}$, $\square\in\{+, -,
\times, / , \wedge\}$ and $\ifz{L}{M}{N}$ is the conditional that tests for
zero.

Here $\U$ denotes a set of unitary operations and $\R$ is a phase shift gate
with a parametrized angle. In this article we fix the former to the CNOT and
Hadamard (H) gates, and the latter to the arbitrary rotation gates
$R_{z(\alpha)}$ and $R_{x(\alpha)}$.

For the remaining constants, $0$ and $1$ represent bits, $\new$ is used to
create a qubit, and $\meas$ to measure it. $\star$ is the inhabitant of the
$\unit$ type, and the sequence $M;N$ is used to discard it. Qubits can be
combined via the tensor product $M\otimes N$ with $\Qlet{x^{S_1} \otimes
y^{S_2}}{M}{N}$ as its corresponding destructor.
 
The system supports lists; $\vnil^A$ represents the empty list, $M::N$ the
constructor and $\Qlet{x^S :: y^{\vec S n}}{M}{N}$ acts as the destructor. 
Finally, the term $\Qfor{k^P}{M}{N}$ allows iterating over parameter lists.

The typing system is defined in Figure~\ref {fig:linear-fragment-typing}. We
write $|\Phi|$ for the list of variables in a typing context $\Phi$. The type
$\vec{A}{n}$ represents a vector of known length $n$ of elements of type $A$.

We differentiate between \textit{state contexts} (Noted with $\Gamma$ and
$\Delta$) and \textit{parameter contexts} (Noted with $\Phi$). For our case of
study, parameter contexts consist only of pairs $x:\nat$ or
$x:\vec{\nat}{(n:\nat)}$, since they are the only non-linear types of variables
that we manage. Every other variable falls under the state context. The terms
$\lambda x^S. M$ and $MN$ correspond to the abstraction and application which
will be used for state-typed terms. The analogous constructions for
parameter-typed terms are $\lambda' x^P. M$ and $M@N$.

In this sense we deviate from the original Proto-Quipper-D type system, which
supports a single context decorated with indices. Instead, we use a linear and
non-linear approach similar to the work of Cervesato and
Pfenning\cite{Cervesato1996ALL}.

A key difference between Quipper (and, by extension, Proto-Quipper-D) and
$\lambdaD$ is the approach to defining circuits. In Quipper, circuits are an
intrinsic part of the language and can be operated upon. In our case, the
translation into SZX diagrams will be mediated with a function defined outside
the language.

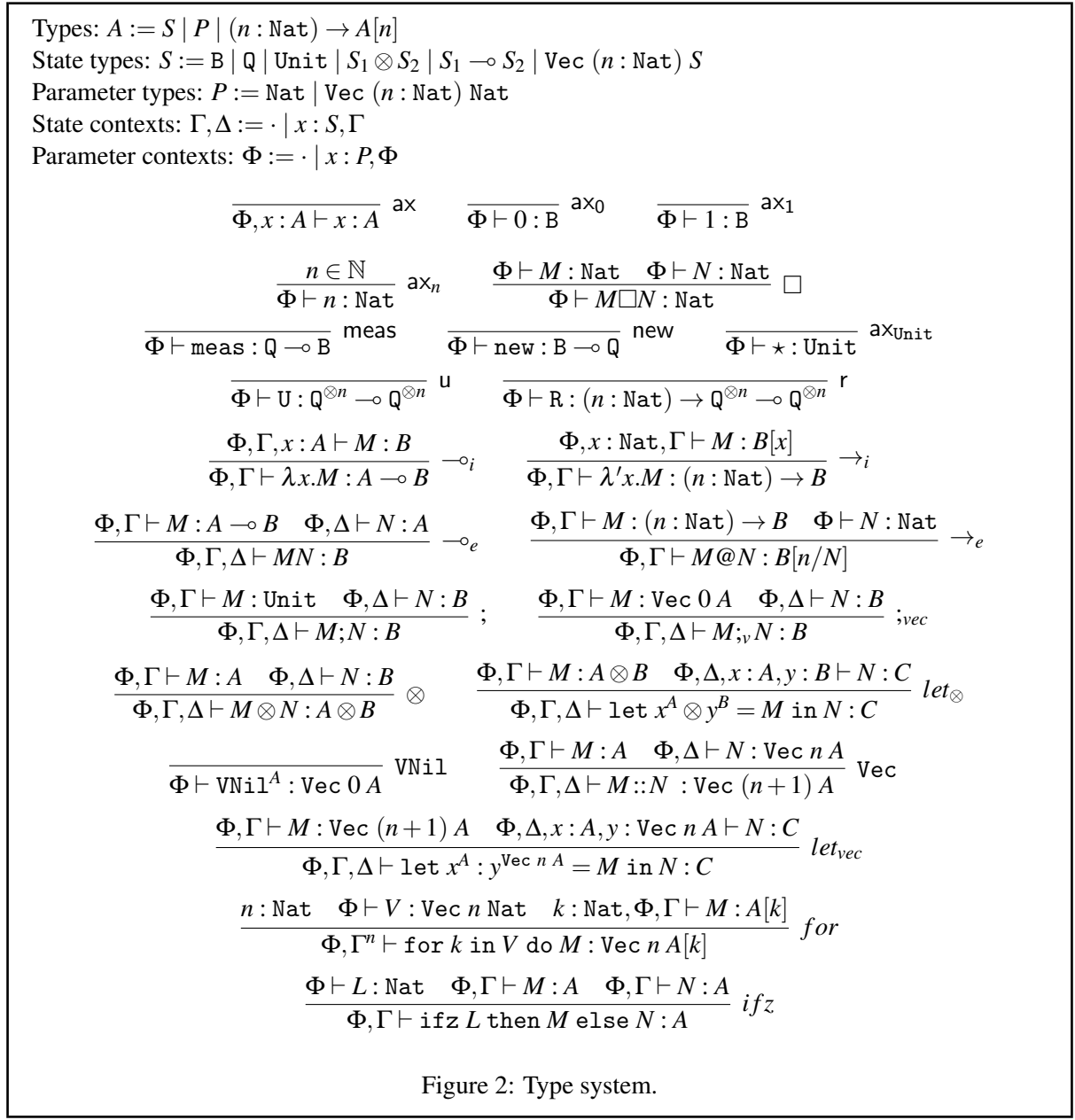
\begin{figure}[hpt]
\begin{mdframed}

    Types:
\(
    A := S \;|\; P \;|\; (n: \nat) \to A[n]
\)

State types:
\(
    S := \bit \;|\; \qubit \;|\;
        \unit \;|\; S_1 \otimes S_2 \;|\; S_1 \multimap S_2 \;|\; \vec{S}{(n: \nat)}
\)

Parameter types:
\(
    P := \nat \;|\; \vec{\nat}{(n: \nat)}
\)

State contexts:
\(
    \Gamma,\Delta := \cdot \;|\; x : S, \Gamma
\)

Parameter contexts:
\(
    \Phi := \cdot \;|\; x : P, \Phi
\)

    \[
        \infer[\mathsf{ax}]{\Phi, x:A\vdash x:A}{} \qquad
        \infer[\mathsf{ax_0}]{\Phi \vdash 0:\bit}{} \qquad
        \infer[\mathsf{ax_1}]{\Phi \vdash 1:\bit}{} \qquad
    \]
    \[
        \infer[\mathsf{ax}_n]{\Phi\vdash n:\nat}{n\in\mathbb{N}}\qquad
        \infer[\square]{\Phi\vdash M\square N:\nat}
        {\Phi\vdash M:\nat & \Phi\vdash N:\nat}
    \]
    \[
        \infer[\mathsf{meas}]{\Phi \vdash \meas:\qubit\multimap\bit}{} \qquad
        \infer[\mathsf{new}]{\Phi \vdash \new:\bit\multimap\qubit}{} \qquad
        \infer[\mathsf{ax_\unit}]{\Phi\vdash \star: \unit}{}
    \]
    \[
        \infer[\mathsf{u}]{\Phi\vdash \U:\qubit^{\otimes n}\multimap\qubit^{\otimes n}}
        {}
        \qquad
        \infer[\mathsf{r}]{\Phi\vdash \R:(n:\nat)\to\qubit^{\otimes n}\multimap\qubit^{\otimes n}}
        {}
    \]
    \[
        \infer[\multimap_i]{\Phi,\Gamma\vdash\lambda x.M:A\multimap B}{\Phi,\Gamma,x:A\vdash M:B}
        \qquad
        \infer[\rightarrow_i]{\Phi,\Gamma\vdash\lambda' x.M:(n:\nat)\to B}{\Phi,x:\nat,\Gamma\vdash M:B[x]}
    \]
    \[
        \infer[\multimap_e]{\Phi,\Gamma,\Delta\vdash MN:B}{\Phi,\Gamma\vdash M:A\multimap B & \Phi,\Delta\vdash N:A}
        \qquad
        \infer[\rightarrow_e]{\Phi,\Gamma\vdash M @ N:B[n/N]}{\Phi,\Gamma\vdash M:(n:\nat)\to B & \Phi\vdash N:\nat}
    \]
    \[
        \infer[;]{\Phi,\Gamma,\Delta\vdash M;N :B}{\Phi,\Gamma\vdash M:\unit & \Phi,\Delta\vdash N:B}
        \qquad
        \infer[;_{vec}]{\Phi,\Gamma,\Delta\vdash M;_vN :B}{\Phi,\Gamma\vdash M:\vec{A}{0} & \Phi,\Delta\vdash N:B}
    \]
    \[
        \infer[\otimes]{\Phi,\Gamma,\Delta\vdash M\otimes N:A \otimes B}{\Phi,\Gamma\vdash M:A & \Phi,\Delta\vdash N:B}
        \qquad
        \infer[let_\otimes]{\Phi,\Gamma,\Delta\vdash \Qlet{x^A \otimes y^B}{M}{N}:C}{\Phi,\Gamma\vdash M:A\otimes B & \Phi,\Delta,x:A,y:B\vdash N: C}
    \]
    \[
        \infer[\vnil]{\Phi\vdash \vnil^A:\vec A 0}{}
        \qquad
        \infer[\vec{}{}]{\Phi,\Gamma,\Delta\vdash M\!::\!N\ :\vec{A}{(n+1)}}{\Phi,\Gamma\vdash M:A & \Phi,\Delta\vdash N: \vec A n}
    \]
    \[
        \infer[let_{vec}]{\Phi,\Gamma,\Delta\vdash \Qlet{x^A : y^{\vec A n}}{M}{N}:C}{\Phi,\Gamma\vdash M:\vec{A}{(n+1)} & \Phi,\Delta,x:A,y:\vec A n\vdash N: C}
    \]
    \[
        \infer[for]{\Phi, \Gamma^n \vdash \Qfor{k}{V}{M} : \vec {A[k]} n}{n:\nat & \Phi\vdash V:\vec \nat n & k:\nat,\Phi,\Gamma \vdash M:A[k]}
    \]
    \[
        \infer[ifz]{\Phi, \Gamma \vdash \ifz{L}{M}{N} : A}{\Phi \vdash L:\nat & \Phi,\Gamma\vdash M:A & \Phi,\Gamma\vdash N: A}
    \]
    \caption{Type system.}%
    \label{fig:linear-fragment-typing}
\end{mdframed}
\end{figure}

Types are divided into two kinds; parameter and state types. Both of these can
depend on terms of type $\nat$. For the scope of this work, this dependence may
only influence the size of vectors types.

Parameter types represent non-linear variable types which are known at the time
of generation of the concrete quantum operations. In the translation into SZX
diagrams, these variables may dictate the labels of the wires and spiders.
Vectors of $\nat$ terms represent their cartesian product.
On the other hand, state types correspond to the quantum operations and states
to be computed. In the translation, these terms inform the shape and composition
of the diagrams. Vectors of state type terms represent their tensor product.

In lieu of unbounded and implicit recursion, we define a series of primitive
functions for performing explicit vector manipulation. These primitives can be
defined in the original language, with the advantage of them being strongly
normalizing. The first four primitives are used to manage state vectors, while
the last one is used for generating parameters. For ease of translation some
terms are decorated with type annotations, however we will omit these for
clarity when the type is apparent.

\begin{center}
\begin{tabular}{l}
    $\Phi\vdash \Qaccumap_{A,B,C} : (n:\nat) \to \vec A n \multimap\; \vec{(A \multimap C \multimap B \otimes C)}{n} \multimap C \multimap (\vec B n) \otimes C $\\
    $\Phi\vdash \Qsplit_A : (n:\nat) \to (m:\nat) \to \vec{A}{(n+m)} \multimap \vec A n \otimes \vec A m $\\
    $\Phi\vdash \Qappend_A : (n:\nat) \to (m:\nat) \to \vec{A}{n} \multimap \vec A m \multimap \vec{A}{(n+m)} $\\
    $\Phi\vdash \Qdrop : (n:\nat) \to \vec{\unit}{n} \multimap\; \unit $\\
    $\Phi\vdash \Qrange : (n:\nat)\to (m:\nat)\to \vec{\nat}{(m-n)}$ \\
\end{tabular}
\end{center}

Since every diagram represents a linear map between qubits there is no
representation equivalent to non-terminating terms, even for weakly normalizing
programs. This is the main reason behind the design choice of the primitives
set.
\reducedorlong{Cf. the long version of this paper on ArXiV for the operational
semantics of the calculus and primitives, as well as the encoding of the
primitives as Proto-Quipper-D functions.}{We include the operational semantics
of the calculus and primitives in Appendix~\ref{sec:op-semantics}. The encoding
of the primitives as Proto-Quipper-D functions is shown in
Appendix~\ref{sec:primitive-trans}. }

We additionally define the following helpful terms based on the previous
primitives to aid in the manipulation of vectors.
\reducedorlong{}{Cf. Appendix~\ref{sec:op-semantics} for their definition as \lambdaD-terms.}
\begin{center}
\begin{tabular}{l}
    $\Phi\vdash \Qmap_{A,B} : (n:\nat) \to \vec A n \multimap\; \vec{(A \multimap B)}{n} \multimap \vec B n $\\
    $\Phi\vdash \Qfold_{A,C} : (n:\nat) \to \vec A n \multimap\; \vec{(A \multimap C \multimap C)}{n} \multimap C \multimap C $\\
    $\Phi\vdash \Qcompose_{A} : (n:\nat) \to \vec{(A \multimap A)}{n} \multimap A \multimap A $\\
\end{tabular}
\end{center}

The distinction between primitives that deal with state and parameters highlights the inclusion of the $\mathtt{for}$ as a construction into the language instead of a primitive. Since it acts over both parameter and state types, its function is effectively to bridge the gap between the two of them. This operation closely corresponds to the list instantiation procedure presented in the Section~\ref{sec:szx}.

For example, if we take $ns$ to be a vector of natural numbers, and $xs$ a vector of abstractions $R @k (new 0)$. The term $\Qfor{k}{ns}{xs}$ generates a vector of quantum maps by instantiating the abstractions for each individual parameter in $ns$.

\section{Encoding programs as diagram families}%
\label{sec:translation}

In this section we introduce an encoding of the lambda calculus presented in Section~\ref{sec:fragment}
into families of SZX diagrams with context variables as inputs and term values as outputs.
We split the lambda-terms into those that represent linear mappings between quantum states
and can be encoded as families of SZX diagrams, and parameter terms that can be completely evaluated
at compile-time.

\subsection{Parameter evaluation}

We say a type is \textit{evaluable} if it has the form
$A = (n_1:\nat) \to \dots \to (n_k:\nat) \to P[n_1,\dots,n_k]$ with $P$ a parameter type.
Since $A$ does not encode a quantum operation,
we interpret it directly into functions over vectors of natural numbers.
The translation of an evaluable type, $\eval{A}$, is defined recursively as follows:
\[
    \eval{(n : \nat) \to B[n]} = \N \to \bigcup_{n \in \N} \eval{B[n]}
    \qquad
    \eval{\nat} = \N
    \qquad
    \eval{\vec{\nat}{(n : \nat)}} = \N^n
\]

Given a type judgement $\Phi \vdash L : P$ where $P$ is an evaluable type,
we define $\eval{L}_\Phi$ as the evaluation of the term into a function from parameters into products of natural numbers.
Since the typing is syntax directed, the evaluation is defined directly over the terms as follows:
\[
    \eval{x}_{x:\nat, \Phi} = x,|\Phi| \mapsto x
    \qquad
    \eval{n}_\Phi = |\Phi| \mapsto n
    \qquad 
    \eval{M \square N}_\Phi = |\Phi| \mapsto \eval{M}_\Phi(|\Phi|) \square \eval{N}_\Phi(|\Phi|)
\]
\[
    \eval{M\ ::\ N}_{\Phi} = |\Phi| \mapsto \eval{M}_\Phi(|\Phi|) \times \eval{N}_\Phi(|\Phi|)
    \qquad 
    \eval{\vnil^\nat}_{\Phi} = |\Phi| \mapsto []
\]
\[
    \eval{\lambda' x^P.M}_\Phi = x,|\Phi| \mapsto \eval{M}_\Phi(x,|\Phi|)
    \qquad
    \eval{M @ N}_\Phi = \eval{M}_\Phi(\eval{N}_\Phi(|\Phi|), \Phi)
\]
\[
    \eval{\ifz{L}{M}{N}}_\Phi = |\Phi| \mapsto \begin{cases}
        \eval{M}_\Phi(|\Phi|) \text{ if } \eval{L}_\Phi(|\Phi|) = 0 \\
        \eval{N}_\Phi(|\Phi|) \text{ otherwise}
    \end{cases}
    \quad
    \eval{\Qrange}_\Phi = n,m,|\Phi| \mapsto \bigtimes_{i=n}^{m-1} i
\]
\[
    \eval{\Qfor{k}{V}{M}}_\Phi = |\Phi| \mapsto \bigtimes_{k\in \eval{V}_\Phi(|\Phi|)} \eval{M}_{k:\nat\Phi}(k, |\Phi|)
\]
\[
    \eval{\Qlet{x^P :: y^{\vec{P}{n}}}{M}{N}}_\Phi = |\Phi|\mapsto 
        \eval{N}_{x:P, y:\vec{P}{n}, \Phi}(y_1,[y_2,\dots,y_n],|\Phi|)
    {\scriptstyle{\text{ where } [y_1,\dots,y_n] = \eval{M}_\Phi(|\Phi|)}}
\]

\begin{lemma}%
    \label{lem:eval-type}
    Given an evaluable type $A$ and a type judgement $\Phi \vdash L : A$,
    $\eval{L}_\Phi \in \bigtimes_{x:P \in \Phi} \eval{P} \to \eval{A}$.
\end{lemma}

\begin{lemma}%
    \label{lem:eval-reduction}
    Given an evaluable type $A$,
    a type judgement $\Phi \vdash L : A$, and $M \to N$, then $\eval{M}_\Phi = \eval{N}_\Phi$.
\end{lemma}

\subsection{Diagram encoding}

A non-evaluable type has necessarily the form $A = (n_1:\nat) \to \dots \to (n_k:\nat) \to S$,
with $S$ any state type.
We call such types \textit{translatable} since they correspond to terms that
encode quantum operations that can be described as families of diagrams.

We first define a translation $\trans{\cdot}$ from state types
into wire multiplicities as follows.
Notice that due to the symmetries of the SZX diagrams
the linear functions have the same representation as the products.
\[
    \trans\bit = 1
    \qquad
    \trans\qubit = 1
    \qquad
    \trans{\vec{A}{(n: \nat)}} = \trans{A}^{\otimes n}
    \qquad
    \trans{A \otimes B} = \trans A \otimes \trans B
    \qquad
    \trans{A \multimap B} = \trans A \otimes \trans B
\]

Given a translatable type judgement
$\Phi, \Gamma \vdash M : (n_1:\nat) \to \dots \to (n_k:\nat) \to S$
we can encode it as a family of SZX diagrams
\(
  n_1, \dots, n_k, |\Phi| \mapsto \begin{tikzpicture}
	\begin{pgfonlayer}{nodelayer}
		\node [style=box] (0) at (0, 0) {$M(|\Phi|)$};
		\node [style=none] (3) at (3.5, 0) {};
		\node [style=none] (4) at (-3, 0) {};
		\node [style=none] (5) at (2.5, 0.5) {$\llbracket S[|\Phi|] \rrbracket$};
		\node [style=none] (6) at (-2.25, 0.5) {$\llbracket\Gamma\rrbracket$};
	\end{pgfonlayer}
	\begin{pgfonlayer}{edgelayer}
		\draw (0) to (3.center);
		\draw (0) to (4.center);
	\end{pgfonlayer}
\end{tikzpicture}

\).
We will omit the brackets in our diagrams for clarity.
In a similar manner to the evaluation, we define
the translation $\trans{M}_{\Phi,\Gamma}$
recursively on the terms as follows: 

\[
    \trans{x}_{\Phi,x:A} = |\Phi|\mapsto \begin{tikzpicture}
	\begin{pgfonlayer}{nodelayer}
		\node [style=none] (0) at (-1, 0) {};
		\node [style=none] (1) at (1, 0) {};
		\node [style=none] (2) at (0, 0.5) {A};
	\end{pgfonlayer}
	\begin{pgfonlayer}{edgelayer}
		\draw [style=bold edge] (0.center) to (1.center);
	\end{pgfonlayer}
\end{tikzpicture}

    \quad
    \trans{0}_{\Phi} = |\Phi|\mapsto\begin{tikzpicture}
	\begin{pgfonlayer}{nodelayer}
		\node [style=none] (1) at (0.5, 0) {};
		\node [style=none] (14) at (0, 0.5) {$\scriptstyle{\qubit}$};
		\node [style=X dot] (15) at (-0.5, 0) {};
	\end{pgfonlayer}
	\begin{pgfonlayer}{edgelayer}
		\draw [style=bold edge] (15) to (1.center);
	\end{pgfonlayer}
\end{tikzpicture}

    \quad
    \trans{1}_{\Phi} = |\Phi|\mapsto\begin{tikzpicture}
	\begin{pgfonlayer}{nodelayer}
		\node [style=none] (1) at (0.5, 0) {};
		\node [style=none] (14) at (0.5, 0.5) {$\scriptstyle{\qubit}$};
		\node [style=X phase dot] (15) at (-0.5, 0) {$\pi$};
	\end{pgfonlayer}
	\begin{pgfonlayer}{edgelayer}
		\draw [style=bold edge] (15) to (1.center);
	\end{pgfonlayer}
\end{tikzpicture}

    \quad
    \trans{\meas}_{\Phi} = |\Phi|\mapsto\input{judgements/linear/meas.tikz}
\]
\[
    \trans{\new}_{\Phi} = |\Phi|\mapsto\input{judgements/linear/new.tikz}
    \quad
    \trans{U}_{\Phi} = |\Phi|\mapsto\input{judgements/linear/unitary.tikz}
    \quad
    \trans{R}_{\Phi} = n,|\Phi|\mapsto\input{judgements/linear/rotation.tikz}
\]
\[
    \trans{\lambda' x^A . M}_{\Phi,\Gamma} = x,|\Phi| \mapsto \input{judgements/linear/lambdaP.tikz}
    \quad
    \trans{M\ @N}_{\Phi,\Gamma} = |\Phi|\mapsto \input{judgements/linear/applyP.tikz}
\]
\[
    \trans{\lambda x^A . M}_{\Phi,\Gamma} = |\Phi| \mapsto \input{judgements/linear/lambda.tikz}
    \quad
    \trans{M\ N}_{\Phi,\Gamma,\Delta} = |\Phi|\mapsto \input{judgements/linear/apply.tikz}
\]
\[
    \trans{M; N}_{\Phi, \Gamma, \Delta} = |\Phi|\mapsto \input{judgements/linear/then.tikz}
    \quad
    \trans{\star}_\Phi = |\Phi|\mapsto\begin{tikzpicture}
	\begin{pgfonlayer}{nodelayer}
		\node [style=bang box, minimum width=1em, minimum height=1em] (0) at (0, 0) {};
	\end{pgfonlayer}
\end{tikzpicture}

    \quad
    \trans{M \otimes N}_{\Phi, \Gamma, \Delta} = |\Phi|\mapsto \input{judgements/linear/product.tikz}
\]
\[
    \trans{M;_v N}_{\Phi, \Gamma, \Delta} = |\Phi|\mapsto \input{judgements/linear/then.tikz}
    \quad
    \trans{\vnil}_\Phi = |\Phi|\mapsto\begin{tikzpicture}
	\begin{pgfonlayer}{nodelayer}
		\node [style=bang box, minimum width=1em, minimum height=1em] (0) at (0, 0) {};
	\end{pgfonlayer}
\end{tikzpicture}

\]
\[
    \trans{let\ x^A \otimes y^B = M\ in\ N}_{\Phi, \Gamma, \Delta} = |\Phi|\mapsto \input{judgements/linear/let.tikz}
\]
\[
    \trans{let\ x^A : y^{\vec{A}{n}} = M\ in\ N}_{\Phi, \Gamma, \Delta} = |\Phi|\mapsto \input{judgements/linear/letVec.tikz}
\]
\[
    \trans{M\!::\!N}_{\Phi, \Gamma, \Delta} = |\Phi|\mapsto \input{judgements/linear/vector.tikz}
    \quad
    \trans{\Qfor{k}{V}{M}}_{\Phi, \Gamma^n} = |\Phi| \mapsto \begin{tikzpicture}
	\begin{pgfonlayer}{nodelayer}
		\node [style=none] (40) at (2.75, 0.5) {$\scriptstyle{A^{\otimes n}}$};
		\node [style=none] (42) at (3.25, 0) {};
		\node [style=none] (43) at (-3, 0) {};
		\node [style=box, inner sep=.75em] (44) at (0, 0) {$\scriptstyle{\substack{\textstyle{M(k)}\\k \in \eval{V}(|\Phi|)}}$};
		\node [style=none] (45) at (-2.5, 0.5) {$\scriptstyle{\Gamma^{\otimes n}}$};
	\end{pgfonlayer}
	\begin{pgfonlayer}{edgelayer}
		\draw [style=bold edge] (43.center) to (44);
		\draw [style=bold edge] (44) to (42.center);
	\end{pgfonlayer}
\end{tikzpicture}

\]
\[
    \trans{\ifz{L}{M}{N}}_{\Phi, \Gamma} = |\Phi| \mapsto \input{judgements/linear/ifz.tikz}
\]
where $\delta$ is the Kronecker delta and $l = \eval{L}(|\Phi|)$. Notice that
the $\new$ and $\meas$ operations share the same translation. Although $\new$
can be encoded as a simple wire, we keep the additional node to maintain the
symmetry with the measurement.

The unitary operators $U$ and rotations $R$ correspond to a predefined set of primitives,
and their translation is defined on a by case basis.
The following table shows the encoding of the operators used in this paper.
\begin{center}
\begin{tabular}{|c|c|c|c|c|c|c|}
    \hline
    Name & Rz(n) & Rz$^{-1}$(n) & Rx(n) & Rx$^{-1}$(n) & H & CNOT \\
    \hline
    Encoding & $\begin{tikzpicture}
	\begin{pgfonlayer}{nodelayer}
		\node [style=none] (1) at (1, 0) {};
		\node [style=none] (2) at (-1, 0) {};
		\node [style=Z phase dot] (3) at (0, 0) {$\frac{\pi}{n}$};
	\end{pgfonlayer}
	\begin{pgfonlayer}{edgelayer}
		\draw (2.center) to (3);
		\draw (3) to (1.center);
	\end{pgfonlayer}
\end{tikzpicture}
$
             & $\begin{tikzpicture}
	\begin{pgfonlayer}{nodelayer}
		\node [style=none] (1) at (1, 0) {};
		\node [style=none] (2) at (-1, 0) {};
		\node [style=Z phase dot, inner sep=.2em] (3) at (0, 0) {$-\frac{\pi}{n}$};
	\end{pgfonlayer}
	\begin{pgfonlayer}{edgelayer}
		\draw (2.center) to (3);
		\draw (3) to (1.center);
	\end{pgfonlayer}
\end{tikzpicture}
$
             & $\begin{tikzpicture}
	\begin{pgfonlayer}{nodelayer}
		\node [style=none] (1) at (1, 0) {};
		\node [style=none] (2) at (-1, 0) {};
		\node [style=X phase dot] (3) at (0, 0) {$\frac{\pi}{n}$};
	\end{pgfonlayer}
	\begin{pgfonlayer}{edgelayer}
		\draw (2.center) to (3);
		\draw (3) to (1.center);
	\end{pgfonlayer}
\end{tikzpicture}
$
             & $\begin{tikzpicture}
	\begin{pgfonlayer}{nodelayer}
		\node [style=none] (1) at (1, 0) {};
		\node [style=none] (2) at (-1, 0) {};
		\node [style=X phase dot, inner sep=.2em] (3) at (0, 0) {$-\frac{\pi}{n}$};
	\end{pgfonlayer}
	\begin{pgfonlayer}{edgelayer}
		\draw (2.center) to (3);
		\draw (3) to (1.center);
	\end{pgfonlayer}
\end{tikzpicture}
$
             & $\begin{tikzpicture}
	\begin{pgfonlayer}{nodelayer}
		\node [style=hadamard] (0) at (0, 0) {};
		\node [style=none] (1) at (1, 0) {};
		\node [style=none] (2) at (-1, 0) {};
	\end{pgfonlayer}
	\begin{pgfonlayer}{edgelayer}
		\draw [in=180, out=0, looseness=0.75] (0) to (1.center);
		\draw [in=180, out=0, looseness=0.75] (2.center) to (0);
	\end{pgfonlayer}
\end{tikzpicture}
$
             & $\begin{tikzpicture}
	\begin{pgfonlayer}{nodelayer}
		\node [style=none] (1) at (1, 0.5) {};
		\node [style=none] (2) at (-1, 0.5) {};
		\node [style=Z dot] (3) at (0, 0.5) {};
		\node [style=X dot] (4) at (0, -0.5) {};
		\node [style=none] (5) at (-1, -0.5) {};
		\node [style=none] (6) at (1, -0.5) {};
	\end{pgfonlayer}
	\begin{pgfonlayer}{edgelayer}
		\draw (2.center) to (3);
		\draw (3) to (1.center);
		\draw (5.center) to (4);
		\draw (4) to (6.center);
		\draw (4) to (3);
	\end{pgfonlayer}
\end{tikzpicture}
$ \\
    \hline
\end{tabular}
\end{center}

The primitives $\Qsplit$, $\Qappend$, $\Qdrop$ and $\Qaccumap$ are translated below. 
Since vectors are isomorphic to products in the wire encoding,
the first three primitives do not perform any operation.
For the accumulating map we utilize the construction presented in Lemma~\ref{lem:accumap},
replacing the function box with a function vector input.
In the latter we omit the wires and gathers connecting the inputs and outputs of the function
to a single wire on the right of the diagram for clarity.
\[
    \trans{\Qsplit_A}_\Phi = n, m, |\Phi| \mapsto \begin{tikzpicture}
	\begin{pgfonlayer}{nodelayer}
		\node [style=none] (40) at (1, 0.25) {$\scriptstyle{(n+m)A \multimap nA \otimes mA}$};
		\node [style=gather] (41) at (-1.5, 0) {};
		\node [style=none] (42) at (3.25, 0) {};
		\node [style=none] (43) at (-2.75, 0.5) {$\scriptstyle{(n+m)A}$};
	\end{pgfonlayer}
	\begin{pgfonlayer}{edgelayer}
		\draw (42.center) to (41.center);
		\draw [in=-135, out=150, loop] (41) to ();
	\end{pgfonlayer}
\end{tikzpicture}

    \quad
    \trans{\Qappend_A}_\Phi = n, m, |\Phi| \mapsto \begin{tikzpicture}
	\begin{pgfonlayer}{nodelayer}
		\node [style=none] (44) at (1.25, 0.25) {$\scriptstyle{nA \multimap mA \multimap (n+m)A}$};
		\node [style=gather] (45) at (-1.5, 0) {};
		\node [style=none] (46) at (3.5, 0) {};
		\node [style=none] (47) at (-2.75, 0.5) {$\scriptstyle{(n+m)A}$};
	\end{pgfonlayer}
	\begin{pgfonlayer}{edgelayer}
		\draw (46.center) to (45.center);
		\draw [in=-135, out=150, loop] (45) to ();
	\end{pgfonlayer}
\end{tikzpicture}

\]
\[
    \trans{\Qdrop}_\Phi = n,|\Phi| \mapsto \begin{tikzpicture}
	\begin{pgfonlayer}{nodelayer}
		\node [style=none] (44) at (0, 0.25) {$\scriptstyle{n 0 \multimap 0}$};
		\node [style=gather] (45) at (-1.5, 0) {};
		\node [style=none] (46) at (.75, 0) {};
	\end{pgfonlayer}
	\begin{pgfonlayer}{edgelayer}
		\draw (46.center) to (45.center);
	\end{pgfonlayer}
\end{tikzpicture}

\]
\[
    \trans{\Qaccumap_{A,B,C}}_\Phi = n, |\Phi| \mapsto \begin{tikzpicture}
	\begin{pgfonlayer}{nodelayer}
		\node [style=split] (0) at (1.5, -3) {};
		\node [style=gather] (1) at (-1, -3) {};
		\node [style=none] (2) at (-8.5, -2.75) {};
		\node [style=none] (3) at (-0.25, -2.75) {$\scriptstyle{nC}$};
		\node [style=none] (5) at (2.25, -3.5) {};
		\node [style=none] (7) at (-5.25, -2.5) {$\scriptstyle{C}$};
		\node [style=none] (8) at (-1.75, -3.5) {};
		\node [style=none] (9) at (2.75, -2.5) {};
		\node [style=none] (11) at (2.5, -2.25) {$\scriptstyle{C}$};
		\node [style=none] (14) at (1, -3.75) {$\scriptstyle{(n-1)C}$};
		\node [style=none] (15) at (1.75, -4) {};
		\node [style=none] (16) at (-1.25, -4) {};
		\node [style=none] (17) at (0.25, -3) {};
		\node [style=none] (19) at (-8.5, 0.25) {};
		\node [style=none] (20) at (2.75, -1.25) {};
		\node [style=none] (21) at (-5.25, 0.5) {$\scriptstyle{nA}$};
		\node [style=none] (22) at (2.5, -1) {$\scriptstyle{nB}$};
		\node [style=split] (23) at (0, -1.5) {};
		\node [style=none] (24) at (-8.5, -1.5) {};
		\node [style=rfarr] (25) at (-5.25, -1.5) {};
		\node [style=none] (26) at (-5.25, -0.5) {$\tau_{n,A,B,C}$};
		\node [style=none] (27) at (-7.25, -1.25) {$\scriptscriptstyle{n (A \multimap C \multimap B \otimes C)}$};
		\node [style=none] (28) at (-2.5, -1.25) {$\scriptscriptstyle{nA \multimap nC \multimap nB \otimes nC}$};
		\node [style=none] (29) at (-2, -2.75) {};
		\node [style=none] (30) at (1, -2) {};
		\node [style=none] (31) at (1.25, -3) {};
		\node [style=none] (32) at (2.25, -1.25) {};
		\node [style=none] (33) at (1.25, -1.75) {};
		\node [style=none] (34) at (-0.25, 0.25) {};
		\node [style=split] (35) at (-11, 1.5) {};
		\node [style=split] (36) at (-11, 0) {};
		\node [style=split] (37) at (-11, -1.5) {};
		\node [style=none] (38) at (-8.5, 4) {};
		\node [style=none] (39) at (-8.5, 2.75) {};
		\node [style=split] (40) at (-3.5, 2.75) {};
		\node [style=split] (41) at (-3.5, 4) {};
		\node [style=none] (42) at (-8.5, 1.75) {};
		\node [style=none] (43) at (3.75, 1.75) {};
		\node [style=none] (44) at (-14.25, 1.5) {};
		\node [style=none] (45) at (-14.25, 0) {};
		\node [style=none] (46) at (-14.25, -1.5) {};
		\node [style=none] (47) at (-12.75, -1) {$\scriptstyle{C}$};
		\node [style=none] (48) at (-12.75, 2) {$\scriptstyle{nA}$};
		\node [style=none] (49) at (-12.75, 0.5) {$\scriptscriptstyle{n (A \multimap C \multimap B \otimes C)}$};
		\node [style=gather] (50) at (3.5, -1.5) {};
		\node [style=gather] (51) at (4.75, 1.5) {};
		\node [style=none] (52) at (7, 1.5) {};
		\node [style=none] (53) at (6, 2) {$\scriptstyle{nB\otimes C}$};
		\node [style=none] (54) at (-2.25, 2.25) {$\scriptstyle{C}$};
		\node [style=none] (55) at (-6.5, 4.5) {$\scriptstyle{0A}$};
		\node [style=none] (56) at (-6.25, 3.25) {$\scriptstyle{0 (A \multimap C \multimap B \otimes C)}$};
		\node [style=none] (58) at (-8.75, 1) {};
		\node [style=none] (59) at (-8.75, -4.75) {};
		\node [style=none] (60) at (4, 1) {};
		\node [style=none] (61) at (4, -4.75) {};
		\node [style=none] (62) at (-7.25, -5.5) {$k\in[0]^{\otimes \delta_{n>0}}$};
		\node [style=none] (63) at (-8.75, 5) {};
		\node [style=none] (64) at (-8.75, 1.25) {};
		\node [style=none] (65) at (-1, 5) {};
		\node [style=none] (66) at (-1, 1.25) {};
		\node [style=none] (67) at (-7.25, 5.75) {$k\in[0]^{\otimes \delta_{n=0}}$};
	\end{pgfonlayer}
	\begin{pgfonlayer}{edgelayer}
		\draw [style=bold edge, in=-150, out=90, looseness=1.25] (8.center) to (1.center);
		\draw [style=bold edge, in=90, out=30, looseness=1.50] (0.center) to (5.center);
		\draw [style=bold edge, in=-135, out=-30, looseness=1.50] (0.center) to (9.center);
		\draw [style=bold edge, in=0, out=-90, looseness=1.75] (5.center) to (15.center);
		\draw [style=bold edge, in=-180, out=-90] (8.center) to (16.center);
		\draw [style=bold edge] (16.center) to (15.center);
		\draw [style=bold edge] (1.center) to (17.center);
		\draw [style=bold edge] (24.center) to (25);
		\draw [style=bold edge] (25) to (23.center);
		\draw [style=bold edge] (2.center) to (29.center);
		\draw [style=bold edge, in=150, out=0] (29.center) to (1.center);
		\draw [style=bold edge, in=90, out=15, looseness=1.75] (23.center) to (30.center);
		\draw [style=bold edge, in=-90, out=0] (17.center) to (30.center);
		\draw [style=bold edge, in=180, out=-45, looseness=0.50] (23.center) to (31.center);
		\draw [style=bold edge] (31.center) to (0.center);
		\draw [style=bold edge, in=180, out=-15] (23.center) to (33.center);
		\draw [style=bold edge, in=-180, out=0, looseness=0.75] (33.center) to (32.center);
		\draw [style=bold edge] (32.center) to (20.center);
		\draw [style=bold edge, in=0, out=45, looseness=1.50] (23.center) to (34.center);
		\draw [style=bold edge] (19.center) to (34.center);
		\draw [style=bold edge] (38.center) to (41.center);
		\draw [style=bold edge] (39.center) to (40.center);
		\draw [style=bold edge] (45.center) to (36.center);
		\draw [style=bold edge] (44.center) to (35.center);
		\draw [style=bold edge] (46.center) to (37.center);
		\draw [style=bold edge, in=180, out=30, looseness=0.75] (35.center) to (38.center);
		\draw [style=bold edge, in=180, out=-30] (35.center) to (19.center);
		\draw [style=bold edge, in=-180, out=30, looseness=0.75] (36.center) to (39.center);
		\draw [style=bold edge, in=-180, out=30, looseness=0.75] (37.center) to (42.center);
		\draw [style=bold edge] (42.center) to (43.center);
		\draw [style=bold edge, in=180, out=-30] (37.center) to (2.center);
		\draw [style=bold edge, in=180, out=-30, looseness=0.75] (36.center) to (24.center);
		\draw [style=bold edge, in=135, out=0] (20.center) to (50.center);
		\draw [style=bold edge, in=-150, out=45] (9.center) to (50.center);
		\draw [style=bold edge, in=-150, out=0, looseness=0.75] (50.center) to (51.center);
		\draw [style=bold edge, in=165, out=0] (43.center) to (51.center);
		\draw [style=bold edge] (51.center) to (52.center);
		\draw [style=dashed edge] (58.center) to (60.center);
		\draw [style=dashed edge] (60.center) to (61.center);
		\draw [style=dashed edge] (61.center) to (59.center);
		\draw [style=dashed edge] (59.center) to (58.center);
		\draw [style=dashed edge] (63.center) to (65.center);
		\draw [style=dashed edge] (65.center) to (66.center);
		\draw [style=dashed edge] (66.center) to (64.center);
		\draw [style=dashed edge] (64.center) to (63.center);
	\end{pgfonlayer}
\end{tikzpicture}

\]
where $\tau_{n,A,B,C}$ is a permutation that rearranges the vectors of functions into tensors of
vectors for each parameter and return value.
That is, $\tau_{n,A,B,C}$ reorders a sequence of registers $(A,C,B,C)\dots(A,C,B,C)$
into the sequence $(A\dots A)(C\dots C)(B\dots B)(C\dots C)$.
It is defined as follows,
\[
    \tau_{n,A,B,C}(i) = \begin{cases}
        \imod{i}{k} + a * (\idiv{i}{k}) & \text{if } \imod{i}{k} < a \\
        \imod{i}{k} + c * (\idiv{i}{k}) + a*(n-1) & \text{if } a \leq \imod{i}{k} < (a+c) \\
        \imod{i}{k} + b * (\idiv{i}{k}) + (a+c)*(n-1) & \text{if } (a+c) \leq \imod{i}{k} < (a+c+b) \\
        \imod{i}{k} + c * (\idiv{i}{k}) + (a+c+b)*(n-1) & \text{if } (a+c+b) \leq \imod{i}{k} \\
    \end{cases}
\]
for $i \in [0,(a+c+b+c)*n)$,
where mod and div are the integer modulo and division operators,
$a = \trans{A}$, $b = \trans{B}$, $c = \trans{C}$,
and $k = a+c+b+c$.

As a consequence of Lemma~\ref{lem:list-instantiation-linear},
the number of nodes in the produced diagrams grows linearly
with the size of the input.
Notice that the ZX spiders, the ground, and the Hadamard operator
are only produced in the translations of the quantum primitives.
We may instead have used other variations of the calculus
supporting the scalable extension, such as the ZH calculus~\cite{Backens_2019},
better suited for other sets of quantum operators.

\begin{lemma}%
    \label{lem:trans-reduction}
    The translation procedure is correct in respect to the operational semantics of \lambdaD.
    If $A$ is a translatable type, $\Phi, \Gamma \vdash M : A$, and $M \to N$,
    then $\trans{M}_{\Phi,\Gamma} = \trans{N}_{\Phi,\Gamma}$.
\end{lemma}
\section{Application example: QFT}%
\label{sec:qft-example}%

The Quantum Fourier Transform is an algorithm used extensively in quantum computation,
notably as part of Shor's algorithm for integer factorization~\cite{nielsen_chuang_quantum_programming}.
The QFT function operates generically over $n$-qubit states and in general a circuit encoding of it
requires $\bigo(n^2)$ gates.
In this section we present an encoding of the algorithm as a \lambdaD\ term,
followed by the translation into a family of constant-sized diagrams.
\reducedorlong{}{The corresponding Proto-Quipper-D program is listed in Appendix~\ref{sec:qft-code}.}

The following presentation divides the algorithm into three parts.
The \textit{crot} term applies a controlled rotation over a qubit with a parametrized angle.
\textit{apply\_crot} operates over the last $n-k$ qubits of an $n$-qubit state
by applying a Hadamard gate to the first one and then using it as target
of successive \textit{crot} applications using the rest of the qubits as controls.
Finally, \textit{qft} repeats \textit{apply\_crot} for all values of $k$.
In the terms, we use $n\dots m$ as a shorthand for $\Qrange\ @n\ @m$.
\\

{
   \setlength{\parindent}{0pt}

   \(
      \text{crot}: (n:\nat)\to (\qubit\otimes\qubit)\multimap (\qubit\otimes\qubit)
   \)\\
   \(
     \text{crot}:= \lambda' n^{\nat}.\lambda qs^{\qubit \otimes \qubit}.
       \Qlet{c^\qubit \otimes q^\qubit}{qs}{}
       \Qlet{c^\qubit \otimes q^\qubit}{\mathtt{CNOT}\ c\ (\mathtt{Rz}\ @2^n\ q)}{}
       \mathtt{CNOT}\ c\ (\mathtt{Rz^{-1}}\ @2^n\ q)
   \)\\

   \(
      \text{apply\_crot}: (n:\nat) \to (k:\nat) \to \vec{\qubit}{n} \multimap \vec{\qubit}{n}
   \)%
   \begin{flalign*}
      \text{apply\_crot}:=\ & \lambda' n^\nat.\ \lambda' k^\nat.\ \lambda qs^{\vec{\qubit}{n}}. & \\
      & \ifz{(n-k)}{qs}{}\\
      & \Qlet{h^{\vec{\qubit}{k}}\otimes qs'^{\ \vec{\qubit}{n-k}}}
         {\Qsplit\ @k\ @(n-k)\ qs}{}\\
      & \Qlet{q^\qubit \otimes cs^{\vec{\qubit}{n-k-1}}}
         {qs'}{}\\
      & \Qlet{fs^{\vec{(\qubit\otimes\qubit\multimap\qubit\otimes\qubit)}{(n-k-1)}}}
         {\Qfor{m^\nat}{2..(n-k+1)}{\text{crot }@m}}{\!\!}\\
      & \Qlet{cs'^{\ (\vec{\qubit}{n-k-1})}\otimes q'^{\ \qubit}}
         {\Qaccumap\ fs\ (H\ q)\ cs}{}\\
      & \text{concat } h\ (q': cs')
   \end{flalign*}

   \(
      \text{qft}: (n:\nat) \to \vec{\qubit}{n}\multimap\vec{\qubit}{n}
   \) %
   \begin{flalign*}
      \text{qft} :=\ & \lambda' n^\nat.\lambda qs^{\vec{\qubit}{n}}.
         \Qcompose & \\
         & (\Qfor{k^\nat}{\text{reverse\_vec } @(0..n)}
         {\lambda qs'^{\ \vec{\qubit}{n}}.\text{apply\_crot } @n\ @k\ qs'})\ qs
   \end{flalign*}
}

The translation of each term into a family of diagrams is shown below.
We omit the wire connecting the function inputs to the right side
of the graphs for clarity and eliminate superfluous gathers and splitters
using rules $\mathbf{(sg)}$ and $\mathbf{(gs)}$.
Notice that, in contrast to a quantum circuit encoding,
the resulting diagram's size does not depend on the number of qubits $n$.

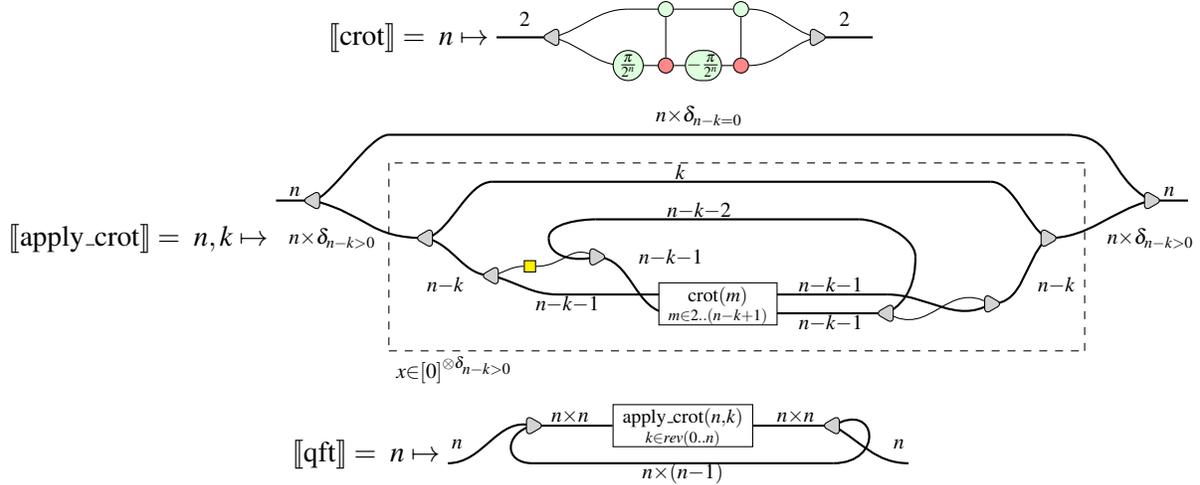
\[
   \trans{\text{crot}}=\ n \mapsto \begin{tikzpicture}
	\begin{pgfonlayer}{nodelayer}
		\node [style=none] (3) at (3.5, 0) {};
		\node [style=none] (7) at (-6.5, 0) {};
		\node [style=gather, rotate=180] (24) at (-5, 0) {};
		\node [style=gather] (25) at (2, 0) {};
		\node [style=none] (29) at (2.75, 0.5) {$\scriptstyle{2}$};
		\node [style=Z dot] (34) at (0, .75) {};
		\node [style=X dot] (35) at (0, -.75) {};
		\node [style=Z dot] (36) at (-2, .75) {};
		\node [style=X dot] (37) at (-2, -.75) {};
		\node [style=Z phase dot] (38) at (-1, -.75) {$-\frac{\pi}{2^n}$};
		\node [style=Z phase dot] (39) at (-3, -.75) {$\frac{\pi}{2^n}$};
		\node [style=none] (40) at (-3, .75) {};
		\node [style=none] (41) at (-5.75, 0.5) {$\scriptstyle{2}$};
	\end{pgfonlayer}
	\begin{pgfonlayer}{edgelayer}
		\draw [in=165, out=0] (34) to (25.center);
		\draw [in=195, out=0] (35) to (25.center);
		\draw [style=bold edge] (25.center) to (3.center);
		\draw (34) to (35);
		\draw (36) to (37);
		\draw [in=180, out=15] (24.center) to (40.center);
		\draw [in=-180, out=-15] (24.center) to (39);
		\draw (39) to (37);
		\draw (40.center) to (36);
		\draw (36) to (34);
		\draw (37) to (38);
		\draw (38) to (35);
		\draw [style=bold edge] (7.center) to (24.center);
	\end{pgfonlayer}
\end{tikzpicture}

\]
\[
   \trans{\text{apply\_crot}}=\ n,k \mapsto \begin{tikzpicture}
	\begin{pgfonlayer}{nodelayer}
		\node [style=none] (3) at (15, 1) {};
		\node [style=none] (7) at (-9.25, 1) {};
		\node [style=gather, rotate=180] (24) at (-5.25, 0) {};
		\node [style=gather] (25) at (11.25, 0) {};
		\node [style=none] (29) at (14.5, 1.25) {$\scriptstyle{n}$};
		\node [style=none] (41) at (-8.75, 1.25) {$\scriptstyle{n}$};
		\node [style=none] (42) at (-3.5, 1.5) {};
		\node [style=none] (43) at (9.75, 1.5) {};
		\node [style=none] (44) at (1.5, 1.75) {$\scriptstyle{k}$};
		\node [style=gather, rotate=180] (45) at (-3.5, -1) {};
		\node [style=none] (46) at (-4.75, -1.25) {$\scriptstyle{n-k}$};
		\node [style=hadamard] (47) at (-2.5, -0.75) {};
		\node [style=gather] (48) at (-0.75, -0.5) {};
		\node [style=none] (49) at (-2, 0) {};
		\node [style=none] (50) at (-0.75, 0.5) {};
		\node [style=box] (51) at (2.5, -1.75) {$\substack{\text{crot}(m)\\\scriptscriptstyle{m \in 2..(n-k+1)}}$};
		\node [style=none] (52) at (1.25, -2) {};
		\node [style=none] (53) at (3.25, -2) {};
		\node [style=gather, rotate=180] (54) at (7, -2) {};
		\node [style=none] (55) at (7.75, -0.5) {};
		\node [style=none] (56) at (6.25, 0.5) {};
		\node [style=none] (57) at (1.25, -1.5) {};
		\node [style=gather] (58) at (9.75, -1.75) {};
		\node [style=none] (59) at (3.25, -1.5) {};
		\node [style=none] (60) at (7, -1.5) {};
		\node [style=none] (61) at (2, 0.75) {$\scriptstyle{n-k-2}$};
		\node [style=none] (62) at (-1.5, -1.75) {$\scriptstyle{n-k-1}$};
		\node [style=none] (63) at (1.25, -0.5) {$\scriptstyle{n-k-1}$};
		\node [style=none] (64) at (5.5, -1.25) {$\scriptstyle{n-k-1}$};
		\node [style=none] (65) at (5.5, -2.25) {$\scriptstyle{n-k-1}$};
		\node [style=none] (66) at (11.5, -1.25) {$\scriptstyle{n-k}$};
		\node [style=none] (67) at (-1.75, -1.5) {};
		\node [style=none] (68) at (-6.25, -3) {};
		\node [style=none] (69) at (-6.25, 2) {};
		\node [style=none] (70) at (12.25, 2) {};
		\node [style=none] (71) at (12.25, -3) {};
		\node [style=gather, rotate=180] (72) at (-8.25, 1) {};
		\node [style=none] (73) at (-6.25, 2.75) {};
		\node [style=none] (74) at (11.75, 2.75) {};
		\node [style=gather] (75) at (14, 1) {};
		\node [style=none] (76) at (2, 3.25) {$\scriptstyle{n\times\delta_{n-k=0}}$};
		\node [style=none] (77) at (14, 0) {$\scriptstyle{n\times\delta_{n-k>0}}$};
		\node [style=none] (78) at (10, -2) {};
		\node [style=none] (79) at (-7.75, 0) {$\scriptstyle{n\times\delta_{n-k>0}}$};
		\node [style=none] (80) at (-4.5, -3.5) {$\scriptstyle{x \in [0]^{\otimes \delta_{n-k>0}}}$};
	\end{pgfonlayer}
	\begin{pgfonlayer}{edgelayer}
		\draw [style=bold edge, in=180, out=15, looseness=0.75] (24.center) to (42.center);
		\draw [style=bold edge] (42.center) to (43.center);
		\draw [style=bold edge, in=150, out=0, looseness=0.75] (43.center) to (25.center);
		\draw [style=bold edge, in=-180, out=-15] (24.center) to (45.center);
		\draw [in=180, out=15] (45.center) to (47);
		\draw [in=150, out=0, looseness=1.25] (47) to (48.center);
		\draw [style=bold edge, in=-165, out=-90] (49.center) to (48.center);
		\draw [style=bold edge, in=-180, out=90, looseness=0.75] (49.center) to (50.center);
		\draw [style=bold edge, in=-90, out=30] (54.center) to (55.center);
		\draw [style=bold edge] (53.center) to (54.center);
		\draw [style=bold edge, in=90, out=0, looseness=0.75] (56.center) to (55.center);
		\draw [style=bold edge] (50.center) to (56.center);
		\draw [style=bold edge] (59.center) to (60.center);
		\draw [style=bold edge, in=-150, out=0] (60.center) to (58.center);
		\draw [style=bold edge, in=225, out=0] (58.center) to (25.center);
		\draw [style=bold edge, in=-180, out=0] (48.center) to (52.center);
		\draw [style=bold edge, in=180, out=-15] (45.center) to (67.center);
		\draw [style=bold edge] (67.center) to (57.center);
		\draw [in=150, out=-30, looseness=1.25] (54.center) to (58.center);
		\draw [style=dashed] (71.center) to (70.center);
		\draw [style=dashed] (70.center) to (69.center);
		\draw [style=dashed] (68.center) to (71.center);
		\draw [style=dashed] (69.center) to (68.center);
		\draw [style=bold edge, in=-180, out=-15, looseness=1.25] (72.center) to (24.center);
		\draw [style=bold edge] (7.center) to (72.center);
		\draw [style=bold edge, in=-180, out=15, looseness=0.50] (72.center) to (73.center);
		\draw [style=bold edge] (74.center) to (73.center);
		\draw [style=bold edge] (75.center) to (3.center);
		\draw [style=bold edge, in=150, out=0, looseness=1.25] (74.center) to (75.center);
		\draw [style=bold edge, in=-165, out=0] (25.center) to (75.center);
	\end{pgfonlayer}
\end{tikzpicture}

\]
\[
   \trans{\text{qft}}=\ n \mapsto \begin{tikzpicture}
	\begin{pgfonlayer}{nodelayer}
		\node [style=none] (3) at (6, -0.25) {};
		\node [style=none] (7) at (-6.25, -0.25) {};
		\node [style=gather, rotate=180] (24) at (4, 0.75) {};
		\node [style=gather] (25) at (-4, 0.75) {};
		\node [style=none] (26) at (-4, -0.25) {};
		\node [style=none] (27) at (4, -0.25) {};
		\node [style=none] (33) at (-6, 0.25) {$\scriptstyle{n}$};
		\node [style=none] (34) at (5.75, 0.25) {$\scriptstyle{n}$};
		\node [style=none] (35) at (0, -0.5) {$\scriptstyle{n\times(n-1)}$};
		\node [style=box] (36) at (0, 0.75) {$\substack{\text{apply\_crot}(n,k)\\\scriptscriptstyle{k \in rev(0..n)}}$};
		\node [style=none] (37) at (-3, 1) {$\scriptstyle{n\times n}$};
		\node [style=none] (38) at (3, 1) {$\scriptstyle{n\times n}$};
	\end{pgfonlayer}
	\begin{pgfonlayer}{edgelayer}
		\draw [style=bold edge, in=165, out=0] (7.center) to (25.center);
		\draw [style=bold edge, in=-180, out=-30] (24.center) to (3.center);
		\draw [style=bold edge, in=0, out=30, looseness=3.50] (24.center) to (27.center);
		\draw [style=bold edge, in=180, out=-165, looseness=2.00] (25.center) to (26.center);
		\draw [style=bold edge] (26.center) to (27.center);
		\draw [style=bold edge] (25.center) to (36);
		\draw [style=bold edge] (36) to (24.center);
	\end{pgfonlayer}
\end{tikzpicture}

\]
\section{Discussion}

In this article, we presented an efficient method to compile
parametric quantum programs written in a fragment of the Proto-Quipper-D
language into families of SZX diagrams.
We restricted the fragment to strongly normalizing terms that can be represented
as diagrams.
Additionally, we introduced a notation to easily compose elements of a diagram
family in parallel.
We proved that our method produces compact diagrams
and shown that it can encode non-trivial algorithms.

A current line of work is defining categorical semantics for the calculus and
families of diagrams, including a subsequent proof of adequacy for the
translation. More work needs to be done to expand the fragment of the Quipper
language that can be translated.

We would like to acknowledge Benoît Valiron for helpful discussion on this
topic, and Frank Fu for his help during the implementation of the
Proto-Quipper-D primitives. This work was supported in part by the French
National Research Agency (ANR) under the research projects SoftQPRO
ANR-17-CE25-0009-02 and Plan France 2030 ANR-22-PETQ-0007, by the DGE of the
French Ministry of Industry under the research project PIA-GDN/QuantEx
P163746-484124, by the project STIC-AmSud project Qapla' 21-SITC-10, the
ECOS-Sud A17C03 project, the PICT-2019-1272 project, the French-Argentinian IRP SINFIN
and the PIP 11220200100368CO project.

\bibliographystyle{eptcs}
\bibliography{main}

\ifreduced
\else

    \newpage
    \appendix

    \section{Semantics of the SZX calculus}%
\label{sec:szx-extended}

We reproduce below the standard interpretation of \szxGND\ diagrams as
density matrices and completely positive
maps~\cite{CJPV19completenessMix,carette_quantum_2021}, modulo scalars.

Let $D_n \subseteq \C^{2^n \times 2^n}$ be the set of n-qubit density
matrices. We define the functor $\interpret{\cdot} : \zxGND \to
\mathbf{CPM(Qubit)}$ which associates to any diagram $\diag : n \to m$ a completely
positive map $\interpret{\diag} : D_n \to D_m$, inductively as
follows.

\[
  \interpret{\diag_1 \otimes \diag_2}
  \;:=\;
  \interpret{\diag_1} \otimes \interpret{D_2}
  \qquad
  \interpret{\diag_2 \circ \diag_1}
  \;:=\;
  \interpret{\diag_2} \circ \interpret{\diag_1}
\]
\[
    \interpret{\input{szx/elem/hadamard-szx.tikz}} :=
      \rho \mapsto V~\rho~V^\dagger
      \text{ where\ }
      V = \sum_{x,y \in \Ftwo^k}
      (-1)^{x \bullet y} \ketbra{y}{x}
\]
\[
    \interpret{\input{szx/elem/spiderZ-szx.tikz}} :=  
      \rho \mapsto V~\rho~V^\dagger
      \text{ where\ }
      V = \sum_{x \in \Ftwo^k}
      e^{i x\bullet\overrightarrow\alpha} \ket{x}^{\otimes m}\bra{x}^{\otimes n}
\]
\[
    \interpret{\input{szx/elem/spiderX-szx.tikz}} :=
    \interpret{\input{szx/elem/hadamard-szx.tikz}}^{\otimes m} \circ
    \interpret{\input{szx/elem/spiderZ-szx.tikz}} \circ
    \interpret{\input{szx/elem/hadamard-szx.tikz}}^{\otimes n}
\]
\[
    \interpret{\begin{tikzpicture}
	\begin{pgfonlayer}{nodelayer}
		\node [style=none] (2) at (-1, 0) {};
		\node [style=none] (3) at (-0.5, 0.25) {$\scriptstyle{k}$};
		\node [style=ground bold, rotate=90] (4) at (0, 0) {$\ground$};
	\end{pgfonlayer}
	\begin{pgfonlayer}{edgelayer}
		\draw [style=bold edge] (2.center) to (4.center);
	\end{pgfonlayer}
\end{tikzpicture}
} :=
      \rho \mapsto
      \sum_{x \in \Ftwo^k}
      \bra{x} \rho \ket{x}
    \qquad
    \interpret{\begin{tikzpicture}
	\begin{pgfonlayer}{nodelayer}
		\node [style=none] (2) at (1, 0) {};
		\node [style=none] (3) at (0.5, 0.25) {$\scriptstyle{k}$};
		\node [style=ground bold, rotate=-90] (4) at (0, 0) {$\ground$};
	\end{pgfonlayer}
	\begin{pgfonlayer}{edgelayer}
		\draw [style=bold edge] (2.center) to (4.center);
	\end{pgfonlayer}
\end{tikzpicture}
} :=
      \sum_{x \in \Ftwo^k}
      \ketbra{x}{x}
    \qquad
    \interpret{\begin{tikzpicture}
	\begin{pgfonlayer}{nodelayer}
		\node [style=none] (1) at (2, 0) {};
		\node [style=none] (2) at (0, 0) {};
		\node [style=none] (3) at (1, 0.25) {$\scriptstyle{k}$};
	\end{pgfonlayer}
	\begin{pgfonlayer}{edgelayer}
		\draw [style=bold edge, in=180, out=0] (2.center) to (1.center);
	\end{pgfonlayer}
\end{tikzpicture}
} :=
    \rho \mapsto \rho
\]
\[
    \interpret{\input{szx/elem/gather.tikz}} :=
      \rho \mapsto \rho
    \qquad
    \interpret{} :=  
      \rho \mapsto V~\rho~V^\dagger
      \text{ where\ }
      V = \sum_{x \in \Ftwo^k}
      \ketbra{\sigma(x)}{x}
\]
\[
    \interpret{\begin{tikzpicture}
	\begin{pgfonlayer}{nodelayer}
		\node [style=none] (1) at (0, 0.5) {};
		\node [style=none] (2) at (0, -0.5) {};
		\node [style=none] (3) at (0.75, 0) {$\scriptstyle{k}$};
	\end{pgfonlayer}
	\begin{pgfonlayer}{edgelayer}
		\draw [style=bold edge, in=0, out=0, looseness=1.75] (2.center) to (1.center);
	\end{pgfonlayer}
\end{tikzpicture}
} :=
      \rho \mapsto
      \sum_{x \in \Ftwo^k}
      \bra{xx} \rho \ket{xx}
    \qquad
    \interpret{\begin{tikzpicture}
	\begin{pgfonlayer}{nodelayer}
		\node [style=none] (1) at (0, 0.5) {};
		\node [style=none] (2) at (0, -0.5) {};
		\node [style=none] (3) at (-0.75, 0) {$\scriptstyle{k}$};
	\end{pgfonlayer}
	\begin{pgfonlayer}{edgelayer}
		\draw [style=bold edge, in=180, out=180, looseness=1.75] (2.center) to (1.center);
	\end{pgfonlayer}
\end{tikzpicture}
} :=
      \sum_{x \in \Ftwo^k}
      \ketbra{xx}{xx}
    \qquad
    \interpret{\begin{tikzpicture}
	\begin{pgfonlayer}{nodelayer}
		\node [style=bang box, minimum width=1em, minimum height=1em] (0) at (0, 0) {};
	\end{pgfonlayer}
\end{tikzpicture}
} := Id_0
\]
\[
    \interpret{\input{szx/elem/swap-szx.tikz}} :=
      \rho \mapsto V~\rho~V^\dagger
      \text{ where\ }
      V = \sum_{x \in \Ftwo^k, y \in \Ftwo^l}
      \ketbra{yx}{xy}
\]
where $\forall u,v\in \mathbb{R}^n, u \bullet v = \sum_{i=1}^m u_i v_i$.

The \szxGND\ calculus defines a set of rewrite rules, shown below.
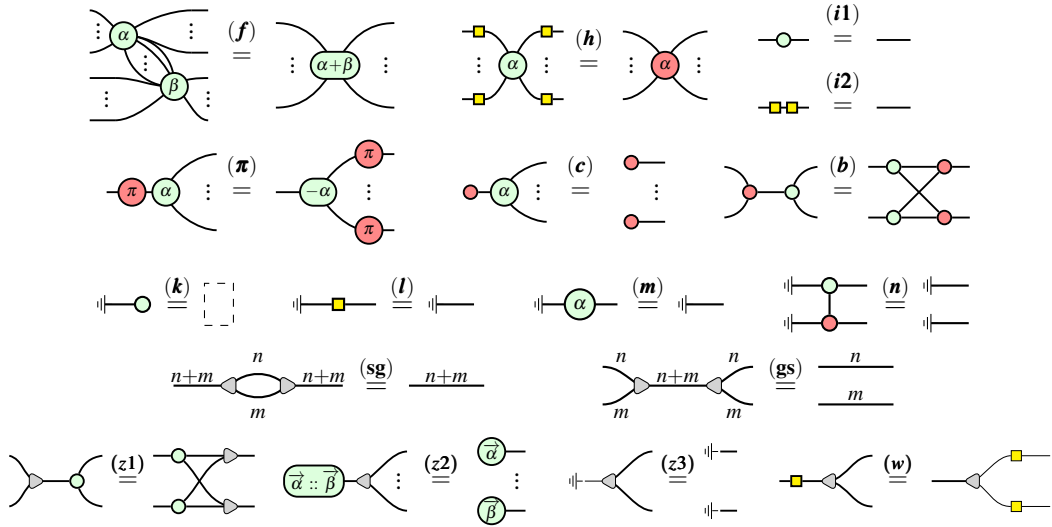
\[\scalebox{0.9}{\input{szx/rules/zxRules.tikz}}\]
\[
  \begin{tikzpicture}
	\begin{pgfonlayer}{nodelayer}
		\node [style=ground bold,rotate=-90] (0) at (0, 0) {\ground};
		\node [style=Z dot bold] (1) at (1, 0) {};
	\end{pgfonlayer}
	\begin{pgfonlayer}{edgelayer}
		\draw [style=bold edge, out=0,in=180,looseness=1.50] (0.center) to (1.center);
	\end{pgfonlayer}
\end{tikzpicture}

  \;\overset{\mbox{\scriptsize\GndRemoveRule}}{=}\;
  \begin{tikzpicture}
	\begin{pgfonlayer}{nodelayer}
		\node [style=bang box, minimum width=1em, minimum height=1em] (0) at (0, 0) {};
	\end{pgfonlayer}
\end{tikzpicture}

  \qquad
  \input{szx/rules/HRule-a.tikz}
  \;\overset{\mbox{\scriptsize\GndHadamardRule}}{=}\;
  \begin{tikzpicture}
	\begin{pgfonlayer}{nodelayer}
		\node [style=ground bold, rotate=-90] (1) at (0, 0) {\ground};
		\node [style=none] (2) at (1, 0) {};
	\end{pgfonlayer}
	\begin{pgfonlayer}{edgelayer}
		\draw [style=bold edge, out=0,in=180,looseness=1.50] (1.center) to (2.center);
	\end{pgfonlayer}
\end{tikzpicture}

  \qquad
  \input{szx/rules/MergeRule-a.tikz}
  \;\overset{\mbox{\scriptsize\GndDiscardRule}}{=}\;
  \begin{tikzpicture}
	\begin{pgfonlayer}{nodelayer}
		\node [style=ground bold, rotate=-90] (1) at (0, 0) {\ground};
		\node [style=none] (2) at (1, 0) {};
	\end{pgfonlayer}
	\begin{pgfonlayer}{edgelayer}
		\draw [style=bold edge, out=0,in=180,looseness=1.50] (1.center) to (2.center);
	\end{pgfonlayer}
\end{tikzpicture}

  \qquad
  \input{szx/rules/CNOTRule-a.tikz}
  \;\overset{\mbox{\scriptsize\GndCNOTRule}}{=}\;
  \input{szx/rules/CNOTRule-b.tikz}
\]
\[
    \input{szx/rules/split-gather.tikz}
    \ \stackrel{\mathbf{(sg)}}{=}\ %
    \begin{tikzpicture}
	\begin{pgfonlayer}{nodelayer}
		\node [style=none] (1) at (2, 0) {};
		\node [style=none] (2) at (0, 0) {};
		\node [style=none] (3) at (1, 0.25) {$\scriptstyle{n+m}$};
	\end{pgfonlayer}
	\begin{pgfonlayer}{edgelayer}
		\draw [style=bold edge, in=180, out=0] (2.center) to (1.center);
	\end{pgfonlayer}
\end{tikzpicture}

    \qquad\qquad
    \input{szx/rules/gather-split.tikz}
    \ \stackrel{\mathbf{(gs)}}{=}\ %
    \input{szx/rules/gather-split-2.tikz}
\]
\[\scalebox{0.9}{\input{szx/rules/gatherRules.tikz}}\]
Additionally, for the arrows restricted to permutations of wires
we have the following rules~\cite{carette_quantum_2021}:
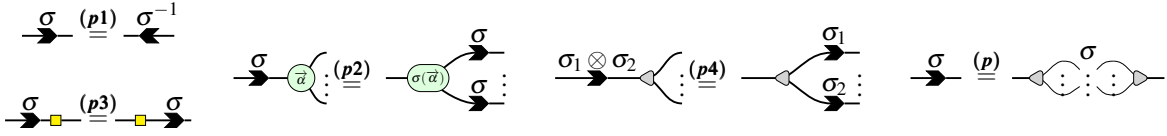
\[\scalebox{0.9}{\input{szx/rules/permRules.tikz}}\]
Finally, since wires with cardinality zero correspond to empty mappings they can be discarded from the diagrams.
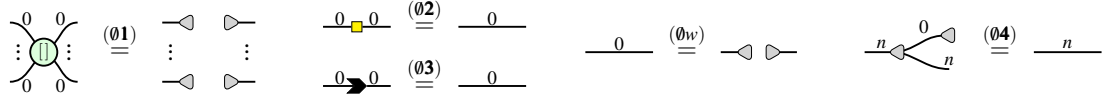
\[\scalebox{0.9}{\input{szx/rules/zeroRules.tikz}}\]
    \section{Operational Semantics of the \texorpdfstring{$\lambdaD$}{lambda sub D} calculus}%
\label{sec:op-semantics}

We define a weak call-by-value small step operational semantics on Table~\ref{fig:rewrite-system}. 
\begin{table}[htbp]
\begin{mdframed}
    \begin{align*}
        V :=\; & x \;|\; C \;|\; 0 \;|\; 1 \;|\; \meas \;|\; \new \;|\; \U \;|\; \\
        & \lambda x^S. M \;|\; \lambda' x^P. M \;| \star \;|\; M \otimes N \;|\; \\
        & \vnil \;|\; M :: N \;\\
    \end{align*}
    \begin{align*}
        (\lambda x.M)V &\to M[V/x]\\
        (\lambda' x.M)@V &\to M[V/x]\\
        \Qlet{x\otimes y}{M_1\otimes M_2}{N}&\to N[x/M_1][y/M_2]\\
        \Qlet{x::y}{M_1 :: M_2}{N}&\to N[x/M_1][y/M_2]\\
        \ifz{V}{M}{N}&\to 
        \begin{cases}
            M\qquad\text{If } V=0 \\ 
            N\qquad\text{Otherwise}
        \end{cases}\\
        \star\ ;\ M &\to M\\
        \vnil\ ;_v\ M &\to M\\
        V_1 \square V_2 &\to V\qquad\text{Where }V_i = n_i\text{ and }V=n_1\square n_2\\
        \Qfor{k}{M_1\ :: M_2}{N} &\to N[k/M_1]\ :: \Qfor{k}{M_2}{N}\\
        \Qfor{k}{\vnil}{N} &\to \vnil\\
    \end{align*}
    \smallskip
    \[
      \infer{MV\to NV}{M\to N}
      \qquad 
      \infer{LM\to LN}{M\to N}
      \qquad
      \infer{M@V\to N@V}{M\to N}
      \qquad 
      \infer{L@M\to L@N}{M\to N}
    \]

    \[
      \infer{\Qlet{x\otimes y}{M}{L}\to \Qlet{x\otimes y}{N}{L}}{M\to N}
      \qquad
      \infer{\Qlet{x:: y}{M}{L}\to \Qlet{x:: y}{N}{L}}{M\to N}
    \]
    \[
      \infer{L\square M\to L\square N}{M\to N}
      \qquad
      \infer{M\square V\to M\square V}{M\to N}
    \]
    \[
      \infer{M\ ;\  L\to N\ ;\ L}{M\to N}
      \qquad
      \infer{\Qlet{x: y}{M}{L}\to \Qlet{x: y}{N}{L}}{M\to N}
    \]
    \[
      \infer{\ifz{M}{L_1}{L_2}\to\ifz{N}{L_1}{L_2}}{M\to N}
      \qquad
      \infer{\Qfor{k}{M}{L}\to\Qfor{k}{N}{L}}{M\to N}
    \]
    \caption{Rewrite system for $\lambdaD$.}
    \label{fig:rewrite-system}
    \end{mdframed}
\end{table}

A key point to note here is that every rewriting rule preserves the state. There are no measurements or unitary operations applied, the rewriting is merely syntactical. Since our goal is translation into an SZX-diagram, this system is powerful enough. We include the rewrite rules for the primitives on Table~\ref{fig:primitives-operation}.

\begin{table}[htbp]
    \begin{mdframed}
        \begin{align*}
        \Qaccumap\ @n\ xs\ fs\ z \to
          & \ifz{n}{xs\ ;_v\ fs\ ;_v\ \vnil\otimes z}{}\\
          & \Qlet{x :: xs'}{xs}{\Qlet{f :: fs'}{fs}{}}\\
          & \Qlet{y\otimes z'}{f\ x\ z}{\Qlet{ys \otimes z''}{\Qaccumap\ @(n-1)\ xs'\ fs'\ z'}{}}\\
          & (y :: ys) \otimes z''\\
        \Qsplit\ @n\ @m\ xs \to
          & \ifz{n}{\vnil\otimes xs}{} \Qlet{y :: xs'}{xs}{}\\
          & \Qlet{ys_1\otimes ys_2}{\Qsplit @(n-1)\ @m\ xs'}{(y :: ys_1)\otimes ys_2}\\
        \Qappend\ @n\ @m\ xs\ ys \to
          & \ifz{n}{xs\ ;_v\ ys}{} \\
          & \Qlet{x :: xs'}{xs}{x :: (\Qappend\ @(n-1)\ @m\ xs'\ ys)}\\
        \Qdrop\ @n\ xs \to
          & \ifz{n}{xs\ ;_v\ \star}{\Qlet{x :: xs'}{xs}{x\ ;\ \Qdrop\ @(n-1)\ xs'}}\\
        \Qrange\ @n\ @m \to
          & \ifz{m-n}{\vnil}{n\ :: \Qrange\ @(n+1)\ @m}\\
        \end{align*}
        \caption{Reductions pertaining to the primitives.}
        \label{fig:primitives-operation}
        \end{mdframed}
    \end{table}

Additionally, we define useful macros based on these functions on Table~\ref{fig:function-macros}. They provide syntactic sugar to deal with state vectors.

\begin{table}[htbp]
  \begin{mdframed}
      \begin{align*}
      \Qmap\ @n\ xs\ fs :=\ 
        & \mathtt{let}\ {fs'\otimes u_1} = \Qaccumap\ @n\ fs\\
        & \qquad (\Qfor{k}{(0..n)}{\lambda f. \lambda u. (\lambda x.\lambda  u. f x \otimes u) \otimes u})\ \star\\
        & \mathtt{in\ let}\ {xs'\otimes u_2} = \Qaccumap\ @n\ xs\ fs'\ \star\\
        & \mathtt{in}\ {u_1\ ;\ u_2\ ;\ xs'}\\
      \Qfold\ @n\ xs\ fs\ z :=\ 
        & \mathtt{let}\ {fs'\otimes u} = \Qaccumap\ @n\ fs\\
        & \qquad (\Qfor{k}{(0..n)}{\lambda f. \lambda u. (\lambda x.\lambda  y. \star \otimes f\  x\  y) \otimes u})\ \star\\
        & \mathtt{in\ let}\ {us\otimes r} = \Qaccumap\ @n\ xs\ fs'\ z\\
        & \mathtt{in}\ {u\ ;\ \Qdrop\ @n\ us\ ;\ r}\\
      \Qcompose\ @n\ xs=\ 
        &\Qfold\ @n\ xs\ (\Qfor{k}{0..n}{(\lambda f.\lambda g.\lambda x. f\ (g\ x))})\ (\lambda x. x)\\
      \end{align*}
      \caption{Function macros.}
      \label{fig:function-macros}
      \end{mdframed}
  \end{table}

    \section{Implementation of primitives in Proto-Quipper-D}%
\label{sec:primitive-trans}

The implicitly recursive primitives defined in Section~\ref{sec:fragment}
can be implemented in proto-quipper-D as follows.
The implementation has been checked with the \texttt{dpq} tool
implemented by Frank Fu (see \url{https://gitlab.com/frank-peng-fu/dpq-remake}).

\inputminted{haskell}{examples/primitives.dpq}

    \section{QFT algorithm in Quipper code}%
\label{sec:qft-code}%

The following Proto-Quipper-D code corresponds to the algorithm
presented in Section~\ref{sec:qft-example}.
This implementation has been checked with the \texttt{dpq} tool
implemented by Frank Fu (see \url{https://gitlab.com/frank-peng-fu/dpq-remake}).
Notice that, in contrast to the presented lambda terms,
the type checker implementation
requires explicit encodings of the Leibniz equalities between parameter types.

\inputminted[lastline=39]{haskell}{examples/qft.dpq}
    \section{Proofs}%
\label{sec:proofs}

\begin{lemmaN}{\!\textbf{(\ref{lem:accumap})}}
    Let $g : 1_n \otimes 1_s \to 1_{m} \otimes 1_s$ and $k\geq 1$, then
    \[
        \input{szx/accumap.tikz}
        \quad=\quad%
        \input{szx/accumap-2.tikz}
    \]
\end{lemmaN}

\begin{proof}
    By induction on $k$.
    If $k=1$,
    \[
        \input{proof/accumap-a-10.tikz}
        \;\overset{\scriptstyle{\bf{(i1, \emptyset 1)}}}{=}\;
        \input{proof/accumap-a-20.tikz}
        \;\overset{\scriptstyle{\bf{(\emptyset 4)}}}{=}\;
        \input{proof/accumap-a-30.tikz}
    \]
    If $k>1$,
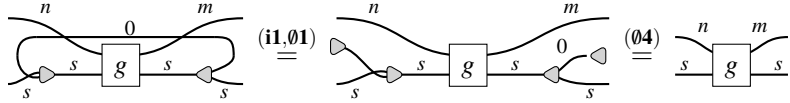
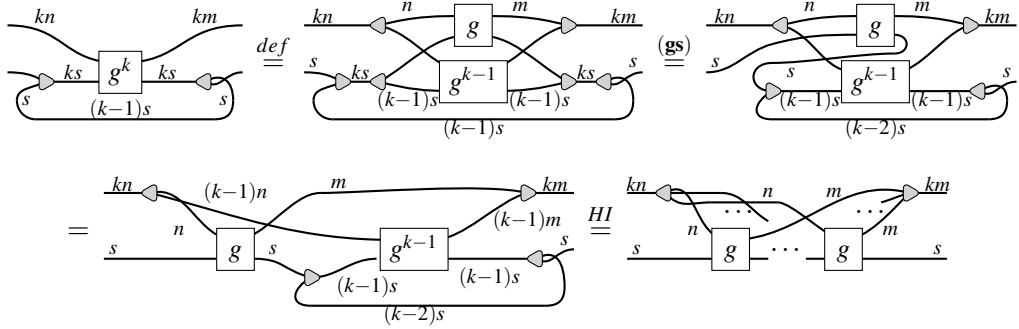
    \[
        \input{proof/accumap-b-10.tikz}
        \;\overset{\scriptstyle{def}}{=}\;
        \input{proof/accumap-b-20.tikz}
        \;\overset{\scriptstyle{\bf{(gs)}}}{=}\;
        \input{proof/accumap-b-30.tikz}
    \]
    \[
        \;=\;
        \input{proof/accumap-b-40.tikz}
        \;\overset{\scriptstyle{HI}}{=}\;
        \input{proof/accumap-b-50.tikz}
    \]
\end{proof}

\begin{lemmaN}{\!\textbf{(\ref{lem:list-instantiation})}}
    For any diagram family $D$, $n_0 : \N$, $N : \N^k$,
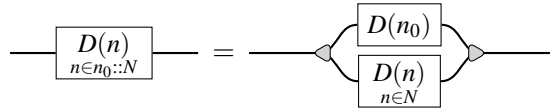
    \[
        \input{szx/list-instantiation-append.tikz}
        \;=\;
        \input{szx/list-instantiation-append-split.tikz}
    \]
\end{lemmaN}

\begin{proof}
    By induction on the term construction
    \begin{itemize}
        \item If $\diag$ is a gather, 
        \[
            \input{proof/list-i-gather-10.tikz}
            \;\overset{\scriptstyle{\bf{(p)}}}{=}\;
            \input{proof/list-i-gather-20.tikz}
        \]
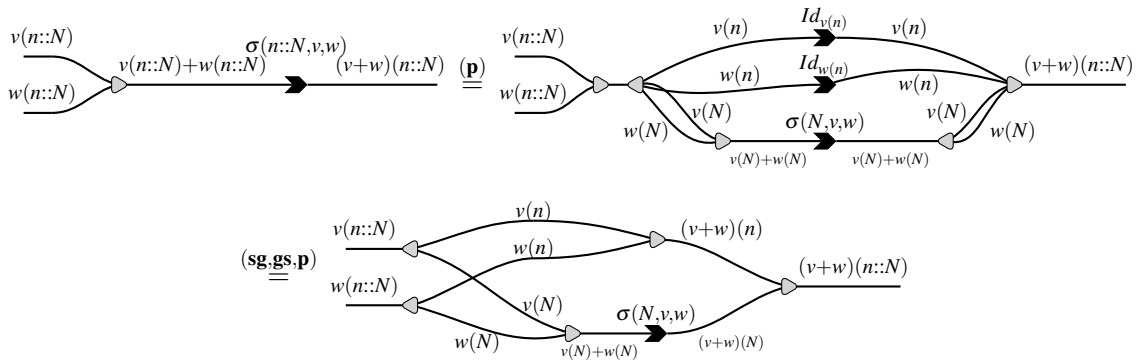
        \[
            \;\overset{\scriptstyle{\bf{(sg,gs,p)}}}{=}\;
            \input{proof/list-i-gather-30.tikz}
        \]
        \item The other cases can be directly derived from the commutation properties
        of the gather generator via rules $(z1), (z2), (z3), (w)$, and $(p4)$.
    \end{itemize}
\end{proof}

\begin{lemmaN}{\!\textbf{(\ref{lem:list-init-zero})}}
    A diagram family initialized with the empty list corresponds to the empty map.
    For any diagram family $D$,
    \[
        \input{szx/list-instantiation-empty.tikz}
        \;=\;
        \input{szx/list-instantiation-empty-flat.tikz}
    \]
\end{lemmaN}

\begin{proof}
    Notice that any wire in the initialized diagrams has cardinality zero.
    By rules ${(\bf\emptyset 1)}$, ${(\bf\emptyset 2)}$, ${(\bf\emptyset 3)}$,
    ${(\bf\emptyset 4)}$, and ${(\bf \emptyset w)}$
    every internal node can be eliminated from the diagram.
\end{proof}

\begin{lemmaN}{\!\textbf{(\ref{lem:list-instantiation-linear})}}
    The list instantiation procedure on an $n$-node diagram family adds
    $\bigo(n)$ nodes to the original diagram.
\end{lemmaN}

\begin{proof}
    By induction on the term construction.
    Notice that the instantiation of any term except the gather
    does not introduce any new nodes, and the gather introduction
    creates exactly one extra node.
    Therefore, the list instantiation adds a number of nodes equal to
    the number of gather generators in the diagram.
\end{proof}

\begin{lemma}%
    \label{lem:eval-substitution}
    Given type judgements $\Phi, x:A \vdash M : B$, and $\Phi\vdash N : A$. $\eval{M}_{x:A,\Phi} (\eval{N}_{\Phi},|\Phi|)= \eval{M[N/x]}_\Phi($ $|\Phi|$ $)$.
\end{lemma}

\begin{proof}
Proof by straightforward induction on $M$.
\end{proof}

\begin{lemmaN}{\!\textbf{(\ref{lem:eval-type})}}
    Given an evaluable type $A$ and a type judgement $\Phi \vdash M : A$,
    $\eval{M}_\Phi \in \bigtimes_{x:P \in \Phi} \eval{P} \to \eval{A}$.
\end{lemmaN}

\begin{proof}
    By induction on the typing judgement $\Phi \vdash M : A$:
    \begin{itemize}
        \item If $\Phi\vdash n:\nat$, then $\eval{n}_\Phi = |\Phi| \mapsto n \in  \bigtimes_{x:P \in \Phi} \eval{P} \to \N$.

        \item If $\Phi, x:A\vdash x:A$, then $\eval{x}_{x:A, \Phi} = x,|\Phi| \mapsto x \in \bigtimes_{y:P \in x:A,\Phi} \eval{P} \to \eval{A}$.

        \item If $\Phi\vdash M\square N:\nat\Phi,\Gamma,\Delta\vdash M\!::\!N\ :\vec{A}{(n+1)}$, then $\eval{M \square N}_\Phi = |\Phi| \mapsto \eval{M}_\Phi(|\Phi|) \square$
        $\eval{N}_\Phi($ $|\Phi|$ $)$. By inductive hypothesis, $\eval{M}_\Phi(|\Phi|), \eval{M}_\Phi(|\Phi|)\in \N$. Then,
        $|\Phi| \mapsto \eval{M}_\Phi(|\Phi|) \square \eval{N}_\Phi(|\Phi|)\in \bigtimes_{x:P \in \Phi} \eval{P} \to \N$.
        
        \item If $\Phi\vdash\lambda' x.M:(x:P)\to B$ then $\eval{\lambda' x^P.M}_\Phi = x,|\Phi| \mapsto \eval{M}_\Phi(x,|\Phi|)$. By inductive hypothesis, $\eval{M}_\Phi(x,|\Phi|)\in \eval{B}$. Then, $|\Phi| \mapsto \eval{M}_\Phi(x,|\Phi|)\in \bigtimes_{y:P \in x:P\Phi} \eval{P} \to \eval{B}$.

        \item If $\Phi,\Gamma\vdash M @ N:B[x/r]$, then $\eval{M @ N}_\Phi = |\Phi| \mapsto\eval{M}_\Phi(\eval{N}_\Phi(|\Phi|), \Phi)$. By inductive hypothesis, $\eval{N}_\Phi(|\Phi|)\in \N$ and $x\mapsto\eval{M}_\Phi(x, \Phi)\in \eval{x:\nat\to B[x]}$. Then, $|\Phi|\mapsto\eval{M}_\Phi(\eval{N}_\Phi(|\Phi|), \Phi) \in \bigtimes_{y:P \in \Phi} \eval{P} \to \eval{B[A/x]}$.

        \item If $\Phi\vdash \vnil^A:\vec A 0$, then $\eval{\vnil^P}_{\Phi} = |\Phi| \mapsto []\in \bigtimes_{x:P \in \Phi} \eval{P} \to \N^0$.

        \item If $\Phi,\Gamma,\Delta\vdash M\!::\!N\ :\vec{A}{(n+1)}$, then $\eval{M\ ::\ N}_{\Phi} = |\Phi| \mapsto \eval{M}_\Phi(|\Phi|) \times \eval{N}_\Phi(|\Phi|)$. By inductive hypothesis ${M}_\Phi(|\Phi|)\in\eval{A}$ and $\eval{N}_\Phi(|\Phi|)\in\eval{A}^n$. Then, $|\Phi| \mapsto \eval{M}_\Phi(|\Phi|) \times \eval{N}_\Phi(|\Phi|)\in\bigtimes_{x:P \in \Phi} \eval{P} \to \eval{A}^{n+1}$.

        \item If $\Phi, \Gamma \vdash \ifz{L}{M}{N} : A$, then \\
        $\eval{\ifz{L}{M}{N}}_\Phi = |\Phi| \mapsto \begin{cases}
            \eval{M}_\Phi(|\Phi|) \text{ if } \eval{L}_\Phi(|\Phi|) = 0 \\
            \eval{N}_\Phi(|\Phi|) \text{ otherwise}
        \end{cases}.$

        By inductive hypothesis $\eval{m}_\Phi(|\Phi|),\eval{N}_\Phi(|\Phi|)\in\eval{A}$ and $\eval{L}_\Phi(|\Phi|)\in\N$. \\
        Then $\eval{\ifz{L}{M}{N}}_\Phi\in\bigtimes_{x:P \in \Phi} \eval{P} \to \eval{A}$.

        \item If $\Phi\vdash \Qfor{k}{N}{M} : \vec A n$, then
        $\eval{\Qfor{k}{V}{M}}_\Phi = |\Phi| \mapsto \bigtimes_{k\in
        \eval{N}_\Phi(|\Phi|)} \eval{M}_{k:\nat\Phi}(k,$ $|\Phi|)$.
        By induction hypothesis, $\eval{N}_\Phi(|\Phi|)\in \nat^n$ and
        $x\mapsto\eval{M}_\Phi(x, \Phi)\in \eval{k:\nat\to A[k]}$. Then $|\Phi|
        \mapsto \bigtimes_{k\in \eval{N}_\Phi(|\Phi|)} \eval{M}_{k:\nat\Phi}(k,
        |\Phi|)\in\bigtimes_{x:P \in \Phi} \eval{P} \to\eval{A[k/N]}^n$.

        \item If $\Phi\vdash \Qlet{x^B : y^{\vec B n}}{M}{N}:A$, then
        $\eval{\Qlet{x^P :: y^{\vec{P}{n}}}{M}{N}}_\Phi =$\allowbreak
        $|\Phi|\mapsto$ \\
        $ \eval{N}_{x:P, y:\vec{P}{n}\Phi}(y_1,[y_2,\dots,y_n],|\Phi|)$
        where $[y_1,\dots,y_n] = \eval{M}_\Phi(|\Phi|)$.
        
        By inductive hypothesis
        $\eval{M}_\Phi(|\Phi|)\in \eval{B}^n$ and $\eval{N}_{x:P, y:\vec{P}{n}
        \Phi}(y_1,[y_2,\dots,y_n],|\Phi|)\in\eval{A}$. Then $|\Phi|\mapsto
        \eval{N}_{x:P, y:\vec{P}{n}
        \Phi}(y_1,[y_2,\dots,y_n],|\Phi|)\in\bigtimes_{x:P \in \Phi} \eval{P}
        \to \eval{A}$.

        \item If $\Phi\vdash \Qrange : (n:\nat)\to (m:\nat)\to \vec{\nat}{(m-n)}$, then $\eval{\Qrange}_\Phi = n,m,|\Phi| \mapsto \bigtimes_{i=n}^{m-1} i\in \bigtimes_{z:P \in x:\nat,y:\nat,\Phi} \eval{P} \to \N^{m-n}$
    \end{itemize}
\end{proof}

\begin{lemmaN}{\!\textbf{(\ref{lem:eval-reduction})}}
    Given an evaluable type $A$,
    a type judgement $\Phi \vdash M : A$, and $M \to N$, then $\eval{M}_\Phi = \eval{N}_\Phi$.
\end{lemmaN}

\begin{proof}
    By induction on the evaluation function $\eval{M}_\Phi$:
    \begin{itemize}
        \item If $M=x$, $M=n$, $M=\vnil^\nat$, $M=\lambda'x. M'$, $M= M_1\ ::\ M_2$ or $M=\Qrange$ then $M$ is in normal form and it does not reduce.
        
        \item If $M = M_1\square M_2$ we have three cases:
        \begin{itemize}
            \item If $M\to M_1\square N$ with $M_2\to N$, then: 
            \begin{align*}
            \eval{M_1 \square M_2}_\Phi &= |\Phi| \mapsto \eval{M_1}_\Phi(|\Phi|) \square \eval{M_2}_\Phi(|\Phi|)\\ 
            &= |\Phi| \mapsto \eval{M_1}_\Phi(|\Phi|) \square \eval{N}_\Phi(|\Phi|)\\
            &= \eval{M_1 \square N}_\Phi\\
            \end{align*}

            \item If $M\to N\square V$ with $M_1\to N$, then: 
            \begin{align*}
            \eval{M_1 \square V}_\Phi &= |\Phi| \mapsto \eval{M_1}_\Phi(|\Phi|) \square \eval{V}_\Phi(|\Phi|)\\ 
            &= |\Phi| \mapsto \eval{N}_\Phi(|\Phi|) \square \eval{V}_\Phi(|\Phi|)\\
            &= \eval{N \square V}_\Phi\\
            \end{align*}

            \item If $M\to n$ with $M_i = n_i\in\mathbb{N}$ and $n= n_1\square n_2$, then: 
            \begin{align*}
            \eval{n_1 \square n_2}_\Phi &= |\Phi| \mapsto \eval{n_1}_\Phi(|\Phi|) \square \eval{n_2}_\Phi(|\Phi|)\\ 
            &= |\Phi| \mapsto n_1\square n_2\\
            &= |\Phi| \mapsto n\\
            &= \eval{n}_\Phi 
            \end{align*}
        \end{itemize}

        \item If $M= M_1\ @\ M_2$ we have three cases:
        \begin{itemize}
            \item If $M\to M_1\ @\ N$ with $M_2\to N$, then: 
            \begin{align*}
            \eval{M_1\ @\ M_2}_\Phi &= |\Phi| \mapsto \eval{M_1}_\Phi(\eval{M_2}_\Phi(|\Phi|), \Phi)\\
            &= |\Phi| \mapsto \eval{M_1}_\Phi(\eval{N}_\Phi(|\Phi|), \Phi)\\
            &= \eval{M_1\ @\ N}_\Phi
            \end{align*}

            \item If $M\to N\ @\ V$ with $M_1\to N$, then: 
            \begin{align*}
            \eval{M_1\ @\ V}_\Phi &= |\Phi| \mapsto \eval{M_1}_\Phi(\eval{V}_\Phi(|\Phi|), \Phi)\\
            &= |\Phi| \mapsto \eval{M_1}_\Phi(\eval{V}_\Phi(|\Phi|), \Phi)\\
            &= \eval{N\ @ V}_\Phi
            \end{align*}

            \item If $M\to M'[V/x]$ with $M_1 = \lambda' x.M'$ and $M_2 = V$, then: 
            \begin{align*}
            \eval{(\lambda' x.M)@ V}_\Phi &= |\Phi| \mapsto \eval{\lambda' x.M}_\Phi(\eval{V}_\Phi(|\Phi|), \Phi)\\
            &= |\Phi| \mapsto (x,|\Phi| \mapsto \eval{M}_{x,\Phi}(x,|\Phi|)) (\eval{V}_\Phi(|\Phi|), \Phi)\\ 
            &= |\Phi| \mapsto \eval{M[V/X]}_\Phi(|\Phi|)\\
            &= \eval{M[V/X]}_\Phi
            \end{align*}
        \end{itemize}

        \item If $M = \ifz{M'}{L}{R}$ we have three cases:
        \begin{itemize}
            \item If $M\to \ifz{N}{L}{R}$ with $M'\to N$, then: 
            \begin{align*}
                \eval{\ifz{M'}{L}{R}}_\Phi &= |\Phi| \mapsto 
                    \begin{cases}
                    \eval{M}_\Phi(|\Phi|) \text{ if } \eval{M'}_\Phi(|\Phi|) = 0 \\
                    \eval{N}_\Phi(|\Phi|) \text{ otherwise}
                    \end{cases}\\
                &= |\Phi| \mapsto 
                \begin{cases}
                \eval{M}_\Phi(|\Phi|) \text{ if } \eval{M'}_\Phi(|\Phi|) = 0 \\
                \eval{N}_\Phi(|\Phi|) \text{ otherwise}
                \end{cases}\\
                &= \eval{\ifz{N}{L}{R}}_\Phi
            \end{align*}

            \item If $M\to L$ with $M'= 0$, then: 
            \begin{align*}
                \eval{\ifz{M'}{L}{R}}_\Phi &= |\Phi| \mapsto 
                \begin{cases}
                \eval{M}_\Phi(|\Phi|) \text{ if } \eval{M'}_\Phi(|\Phi|) = 0 \\
                \eval{N}_\Phi(|\Phi|) \text{ otherwise}
                \end{cases}\\
            &= |\Phi| \mapsto \eval{L}_\Phi(|\Phi|)\\
            &= \eval{L}_\Phi
            \end{align*}

            \item The symmetric case for the else branch is similar to the previous one.
        \end{itemize}

        \item If $M = \Qfor{k}{M'}{R}$ we have three cases:
        \begin{itemize}
            \item If $M\to \Qfor{k}{N}{R}$ with $M'\to N$, then: 
            \begin{align*}
                \eval{\Qfor{k}{M'}{R}}_\Phi &= |\Phi| \mapsto \bigtimes_{k\in \eval{M'}_\Phi(|\Phi|)} \eval{R}_\Phi(k, |\Phi|)\\
                &= |\Phi| \mapsto \bigtimes_{k\in \eval{N}_\Phi(|\Phi|)} \eval{R}_\Phi(k, |\Phi|)\\
                &= \eval{\Qfor{k}{N}{R}}_\Phi
            \end{align*}

            \item If $M\to R[k/M_1]\ : \Qfor{k}{M_2}{R}$ with $M' = V::L$, then: 
            \begin{align*}
                \eval{\Qfor{k}{M_1\ ::\ M_2}{R}}_\Phi &= |\Phi| \mapsto \bigtimes_{k\in \eval{M_1\ ::\ M_2}_\Phi(|\Phi|)} \eval{R}_\Phi(k, |\Phi|)\\
                &= |\Phi| \mapsto \bigtimes_{k\in \eval{M_1}_\Phi(|\Phi|) \times \eval{M_2}_\Phi(|\Phi|)} \eval{R}_\Phi(k, |\Phi|)\\
                &= |\Phi| \mapsto \eval{R}_{k,\Phi}(\eval{M_1}_\Phi(|\Phi|), |\Phi|) \times \bigtimes_{k\in \eval{M_2}_\Phi(|\Phi|} \eval{R}_{k,\Phi}(k, |\Phi|)\\
                &= \eval{R[M_1/k]}_{\Phi} \times \eval{\Qfor{k}{M_2}{R}}_\Phi\\
                &= \eval{R[M_1/k]\ ::\ \Qfor{k}{M_2}{R}}_\Phi\\
            \end{align*}

            \item If $M\to \vnil$ with $M' = \vnil^\nat$, then: 
            \begin{align*}
                \eval{\Qfor{k}{\vnil^\nat}{R}}_\Phi &= |\Phi| \mapsto \bigtimes_{k\in \eval{\vnil^\nat}_\Phi(|\Phi|)} \eval{R}_\Phi(k, |\Phi|)\\
                &= |\Phi| \mapsto \bigtimes_{k\in []} \eval{R}_\Phi(k, |\Phi|)\\
                &= |\Phi| \mapsto []\\
                &= \eval{\vnil^\nat}_\Phi
            \end{align*}
        \end{itemize}

        \item If $M = \Qlet{x::y}{M_1}{M_2}$ we have two cases:
        \begin{itemize}
            \item If $M\to \Qlet{x::y}{N}{M_2}$ with $M_1\to N$, then: 
            \begin{align*}
            \eval{\Qlet{x::y}{M_1}{M_2}}_\Phi &= |\Phi|\mapsto \eval{M_2}_\Phi(y_1,[y_2,\dots,y_n],|\Phi|){\scriptstyle{\text{ where } [y_1,\dots,y_n] = \eval{M_1}_\Phi(|\Phi|)}}\\
            &= |\Phi|\mapsto 
            \eval{M_2}_\Phi(y_1,[y_2,\dots,y_n],|\Phi|)
        {\scriptstyle{\text{ where } [y_1,\dots,y_n] = \eval{N}_\Phi(|\Phi|)}}\\
            &= \eval{\Qlet{x::y}{N}{M_2}}_\Phi
            \end{align*}

            \item If $M\to N[x/M_1][y/M_2]$ with $M' = M_1::M_2$, then: 
            \begin{align*}
                \eval{\Qfor{k}{M_1\ ::\ M_2}{N}}_\Phi &= |\Phi|\mapsto \eval{N}_{x,y,\Phi}(y_1,[y_2,\dots,y_n],|\Phi|)\\
                &= |\Phi|\mapsto \eval{N}_{x,y,\Phi}(M_1, M_2,|\Phi|)\\
                &= |\Phi|\mapsto \eval{N[M_1/x][M_2/y]}_\Phi(|\Phi|)\\
                &= \eval{N[M_1/x][M_2/y]}_\Phi
            \end{align*}
            where $[y_1,\dots,y_n] = \eval{M_1::M_2}_\Phi(|\Phi|)$.
        \end{itemize}

        \item If $M= \Qrange\ @n\ @M_2$ then:
        \begin{align*}
            \eval{\Qrange @n @m}_\Phi &= |\Phi| \mapsto \bigtimes_{i=n}^{m-1} i\\
            &=|\Phi| \mapsto \begin{cases}
                [] \text{ if } n-m = 0 \\
                n\times \bigtimes_{i=n+1}^{m-1} i \text{ otherwise}
            \end{cases}\\
            &=|\Phi| \mapsto \begin{cases}
                [] \text{ if } \eval{n-m}_\Phi = 0 \\
                \eval{n}_\Phi\times \eval{\Qrange @(n+1) @ m)} \text{ otherwise}
            \end{cases}\\
            &= \eval{\ifz{m-n}{\vnil}{n\ :: \Qrange\ @(n+1)\ @m}}_\Phi
        \end{align*}
    \end{itemize}
\end{proof}

\begin{lemmaN}{\!\textbf{(\ref{lem:trans-reduction})}}
    The translation procedure is correct in respect to the operational semantics.
    If $A$ is a translatable type, $\Phi, \Gamma \vdash M : A$, and $M \to N$,
    then $\trans{M}_{\Phi,\Gamma} = \trans{N}_{\Phi,\Gamma}$.
\end{lemmaN}

\begin{proof}
    By case analysis on the reductions of translatable terms.
    \begin{itemize}
        \item If $M = (\lambda x^A.M')V$ and $N=M'[V/x]$,
            \[
                \trans{(\lambda x^A.M')V}_{\Phi,\Delta,\Gamma}
                =
                |\Phi|\mapsto\input{proof/trans-reduct-apply-10.tikz}
            \]
            \[
                \;\overset{\scriptstyle{\bf{(gs)}}}{=}\;
                |\Phi|\mapsto\input{proof/trans-reduct-apply-20.tikz}
                =
                \trans{M'[V/x]}_{\Phi,\Delta,\Gamma}
            \]
        \item If $M = (\lambda' x^A.M')@V$ and $N=M'[V/x]$,
            \[
                \trans{(\lambda' x^A.M')@V}_{\Phi,\Gamma}
                =
                |\Phi|\mapsto\input{proof/trans-reduct-applyP-10.tikz}
            \]
            \[
                =
                |\Phi|\mapsto\input{proof/trans-reduct-applyP-20.tikz}
                =
                \trans{M'[V/x]}_{\Phi,\Gamma}
            \]
        \item If $M = \Qlet{x^A\otimes y^B}{V_1 \otimes V_2}{M'}$ and $N=M'[V_1/x][V_2/y]$,
            \[
                \trans{\Qlet{x^A\otimes y^B}{V_1 \otimes V_2}{M'}}_{\Phi,\Gamma,\Delta,\Lambda}
                =
                |\Phi|\mapsto\input{proof/trans-reduct-let-10.tikz}
            \]
            \[
                \;\overset{\scriptstyle{\bf{(gs)}}}{=}\;
                |\Phi|\mapsto\input{proof/trans-reduct-let-20.tikz}
                =
                \trans{M'[V_1/x][V_2/y]}_{\Phi,\Delta,\Gamma,\Lambda}
            \]
        \item If $M = \Qlet{x^A :: y^{\vec{A}{n}}}{V_1 :: V_2}{M'}$ and $N=M'[V_1/x][V_2/y]$,
            \[
                \trans{\Qlet{x^A :: y^{\vec{A}{n}}}{V_1 :: V_2}{M'}}_{\Phi,\Gamma,\Delta,\Lambda}
                =
                |\Phi|\mapsto\input{proof/trans-reduct-letVec-10.tikz}
            \]
            \[
                \;\overset{\scriptstyle{\bf{(gs)}}}{=}\;
                |\Phi|\mapsto\input{proof/trans-reduct-letVec-20.tikz}
                =
                \trans{M'[V_1/x][V_2/y]}_{\Phi,\Delta,\Gamma,\Lambda}
            \]
        \item If $M = \ifz{L}{M'}{N'}$,
            \begin{itemize}
                \item if $L=0$ and $N = M'$, $\eval{L}_\Phi = 0$ and
                    \[
                        \trans{\ifz{L}{M'}{N'}}_{\Phi,\Gamma}
                        =
                        |\Phi|\mapsto\input{proof/trans-reduct-ifz-10.tikz}
                    \]
                    \[
                        \;\overset{\scriptstyle{Lemma~\ref{lem:list-init-zero}}}{=}\;
                        |\Phi|\mapsto\input{proof/trans-reduct-ifz-20.tikz}
                        \;\overset{\scriptstyle{\bf{(\emptyset 4)}}}{=}\;
                        |\Phi|\mapsto\input{proof/trans-reduct-ifz-30.tikz}
                        =
                        \trans{M'}_{\Phi,\Gamma}
                    \]
                \item The case where $L>0$ and $N = N'$ is symmetric to the case above.
            \end{itemize}
        \item If $M = \vnil;M'$ and $N=M'$,
            \[
                \trans{\vnil;_v M'}_{\Phi,\Gamma}
                =
                |\Phi|\mapsto\input{proof/trans-reduct-then-10.tikz}
                =
                \trans{M'}_{\Phi,\Gamma}
            \]
        \item If $M = \star;M'$ and $N=M'$,
            \[
                \trans{\star;M'}_{\Phi,\Gamma}
                =
                |\Phi|\mapsto\input{proof/trans-reduct-then-10.tikz}
                =
                \trans{M'}_{\Phi,\Gamma}
            \]
        \item If $M = \Qfor{k}{V :: M'}{N'}$ and $N=N'[k/V]::\Qfor{k}{M'}{N'}$,
            \[
                \trans{\Qfor{k}{V :: M'}{N'}}_{\Phi,\Gamma^{\otimes n}}
                =
                |\Phi|\mapsto\input{proof/trans-reduct-for-10.tikz}
            \]
            \[
                \;\overset{\scriptstyle{\text{Lemma~\ref{lem:list-instantiation}}}}{=}\;
                |\Phi|\mapsto\input{proof/trans-reduct-for-20.tikz}
                =
                \trans{N'[k/V]::\Qfor{k}{M'}{N'}}_{\Phi,\Gamma^{\otimes n}}
            \]
        \item If $M = \Qfor{k}{\vnil}{N'}$ and $N=\vnil$,
            \[
                \trans{\Qfor{k}{\vnil}{N'}}_{\Phi}
                =
                |\Phi|\mapsto\input{proof/trans-reduct-forNil-10.tikz}
            \]
            \[
                \;\overset{\scriptstyle{Lemma~\ref{lem:list-init-zero}}}{=}\;
                |\Phi|\mapsto\input{proof/trans-reduct-forNil-20.tikz}
                =
                |\Phi|\mapsto\begin{tikzpicture}
	\begin{pgfonlayer}{nodelayer}
		\node [style=bang box, minimum width=1em, minimum height=1em] (0) at (0, 0) {};
	\end{pgfonlayer}
\end{tikzpicture}

                =
                \trans{\vnil}_{\Phi}
            \]
        \item If $M = \Qaccumap_{A,B,C}\ N'\ M_{xs}\ M_{fs}\ M_{z}$ and \(N=
                \ifz{N'}{xs\ ;_v\ fs\ ;_v\ \vnil\otimes M_z}{}$\\
                $\Qlet{x :: xs'}{M_{xs}}{} \Qlet{f :: fs'}{M_{fs}}{}
                \Qlet{y\otimes z'}{f\ x\ z}{}
                \mathtt{let\ } ys \otimes z'' = \Qaccumap\ @(N'-1)\ xs'\ fs'\ z'
                \mathtt{\ in\ } (y :: ys) \otimes z''
            \). Let $n = \eval{N_1}_\Phi(|\Phi|)$.

            Notice that, by definition of $\tau_{n,A,B,C}$,
            \[
                \input{proof/tau-n.tikz}
                \;\overset{\scriptstyle{{\bf(p)}}}{=}\;
                \input{proof/tau-n1.tikz}
            \]

            Therefore,

            \(
                \trans{\Qaccumap_{A,B,C}\ N'\ M_{xs}\ M_{fs}\ M_{z}}_{\Phi,\Gamma,\Delta,\Pi}
            \)
            \[
                =
                |\Phi|\mapsto\input{proof/trans-reduct-accumap-10.tikz}
            \]
            \[
                \;\overset{\scriptstyle{{\bf(p,gs,sg)}}}{=}\;
                |\Phi|\mapsto\scalebox{0.95}{\input{proof/trans-reduct-accumap-20.tikz}}
            \]
            \(
                =
                \llbracket
                    \ifz{N'}{xs\ ;_v\ fs\ ;_v\ \vnil\otimes M_z}{}
                    \Qlet{x :: xs'}{M_{xs}}{} \Qlet{f :: fs'}{M_{fs}}{}
                    \mathtt{let\ } y\otimes z' = f\ x\ z \mathtt{\ in\ }
                    \Qlet{ys \otimes z''}{\Qaccumap\ @(N'-1)\ xs'\ fs'\ z'}{}
                    (y :: ys) \otimes z''
                \rrbracket_{\Phi,\Gamma,\Delta,\Pi}
            \)

        \item If $M = \Qsplit_A\ @n\ @m\ xs$ and \(N=
            \ifz{n}{\vnil\otimes xs}{}
            \Qlet{y :: xs'}{xs}{}
            \mathtt{let\ } ys_1\otimes ys_2 = \Qsplit @(n-1)\ @m\ xs' \mathtt{\ in\ }
            (y :: ys_1)\otimes ys_2
            \). Let $n = \eval{N_1}_\Phi(|\Phi|)$ and $m = \eval{N_2}_\Phi(|\Phi|)$.
            \[
                \trans{\Qsplit_A\ @n\ @m\ xs}_{\Phi,\Gamma}
                =
                |\Phi|\mapsto\input{proof/trans-reduct-split-10.tikz}
                \;\overset{\scriptstyle{{\bf(sg,gs)}}}{=}\;
                |\Phi|\mapsto\input{proof/trans-reduct-split-20.tikz}
            \]
            \[
                \;\overset{\scriptstyle{{\bf(\emptyset 4)},Lemma~\ref{lem:list-init-zero}}}{=}\;
                |\Phi|\mapsto\input{proof/trans-reduct-split-30.tikz}
            \]
            \[
                \;\overset{\scriptstyle{{\bf(sg,gs,\emptyset 4)}}}{=}\;
                |\Phi|\mapsto\input{proof/trans-reduct-split-40.tikz}
            \]
            \(
                =
                \llbracket
                    \ifz{n}{\vnil\otimes xs}{}
                    \Qlet{y :: xs'}{xs}{}
                    \Qlet{ys_1\otimes ys_2}{\Qsplit @(n-1)\ @m\ xs'}{}
                    (y :: ys_1)\otimes ys_2
                \rrbracket_{\Phi,\Gamma}
            \)
        \item If $M = \Qappend_A\ @N_1\ @N_2\ M_1\ M_2$ and \(N=
            \ifz{N_1}{M_1\ ;_v\ M_2}{}
            \Qlet{x :: xs'}{M_1}{}
            x :: (\Qappend\ @(N_1-1)\ @N_2\ M_1\ M_2)
            \). Let $n = \eval{N_1}_\Phi(|\Phi|)$ and $m = \eval{N_2}_\Phi(|\Phi|)$.
            \[
                \trans{\Qappend_A\ @N_1\ @N_2\ M_1\ M_2}_{\Phi,\Gamma,\Delta}
                =
                |\Phi|\mapsto\input{proof/trans-reduct-append-10.tikz}
            \]
            \[
                \;\overset{\scriptstyle{{\bf(\emptyset 4)},Lemma~\ref{lem:list-init-zero}}}{=}\;
                |\Phi|\mapsto\input{proof/trans-reduct-append-20.tikz}
            \]
            \[
                \;\overset{\scriptstyle{{\bf(sg,gs)}}}{=}\;
                |\Phi|\mapsto\input{proof/trans-reduct-append-30.tikz}
            \]
            \[
                \;\overset{\scriptstyle{{\bf(i1, \emptyset 1, \emptyset 4, sg, gs)}}}{=}\;
                |\Phi|\mapsto\input{proof/trans-reduct-append-40.tikz}
            \]
            \(
                =
                \left\llbracket\ifz{N_1}{M_1\ ;_v\ M_2}{}
                    \Qlet{x :: xs'}{M_1}{}
                    x :: (\Qappend\ @(N_1-1)\ @N_2\ M_1\ M_2\right
                \rrbracket_{\Phi,\Gamma,\Delta}
            \)
        \item If $M = \Qdrop\ @N'\ M'$ and \(N=
            \ifz{N'}{M'\ ;\ \star}{} \Qlet{x :: xs'}{M'}{} x\ ;\ \Qdrop\ @(N'-1)\ xs'
            \). Let $l = \eval{N'}_\Phi(|\Phi|)$,
            \[
                \trans{\Qdrop\ @N'\ M'}_{\Phi,\Gamma}
                =
                |\Phi|\mapsto\input{proof/trans-reduct-drop-10.tikz}
                \;\overset{\scriptstyle{{\bf(\emptyset 4)},Lemma~\ref{lem:list-init-zero}}}{=}\;
                |\Phi|\mapsto\input{proof/trans-reduct-drop-20.tikz}
            \]
            \[
                \;\overset{\scriptstyle{{\bf(z3,\emptyset 4)}}}{=}\;
                |\Phi|\mapsto\input{proof/trans-reduct-drop-30.tikz}
            \]
            \[
                =
                \trans{\ifz{N'}{M'\ ;\ \star}{\Qlet{x :: xs'}{M'}{x\ ;\ \Qdrop\ @(N'-1)\ xs'}}}_{\Phi,\Gamma}
            \]
        \item If $M \to N$ is an internal reduction of a translatable term,
            then the diagrams result equivalent via the inductive hypothesis.
        \item If $M \to N$ is an internal reduction of an evaluable term,
            then the diagrams result equivalent via the inductive hypothesis
            and Lemma~\ref{lem:eval-reduction}.
    \end{itemize}
\end{proof}
\fi

\end{document}